\documentclass[12pt,a4paper]{amsart}

\usepackage{amsmath,amsfonts,amssymb,mathtools,extarrows}
\usepackage{xcolor}
\usepackage{comment}

\usepackage{color}
\usepackage{graphicx}

\usepackage{enumerate}  

\usepackage{cancel}

\usepackage[hmargin=2cm,vmargin=2cm]{geometry}

\usepackage[hypertexnames=false,hyperfootnotes=false,colorlinks=true,linkcolor=blue,%
citecolor=purple,filecolor=magenta,urlcolor=cyan,unicode,linktocpage=true,pagebackref=false]{hyperref}
\usepackage{nameref,zref-xr}     

\usepackage{tikz}

\newcommand{\mbP}{\mathbb P}
\newcommand{\mbZ}{\mathbb Z}
\newcommand{\mbC}{\mathbb C}

\newcommand{\cP}{\mathcal P}

\newcommand{\oM}{\overline{\mathcal M}}

\newcommand{\tu}{{\widetilde u}}
\newcommand{\og}{\overline g}
\newcommand{\oh}{\overline h}
\newcommand{\hLambda}{\widehat\Lambda}

\def\cM{{\mathcal{M}}}
\def\oM{{\overline{\mathcal{M}}}}

\renewcommand{\Im}{\mathrm{Im}}
\def\mbQ{{\mathbb Q}}
\def\d{{\partial}}

\renewcommand{\>}{\right>}
\newcommand{\eps}{\varepsilon}

\newcommand{\cA}{\mathcal A}
\newcommand{\hcA}{\widehat{\mathcal A}}
\newcommand{\DR}{\mathrm{DR}}

\newcommand{\ct}{\mathrm{ct}}

\renewcommand{\th}{\widetilde h}

\newcommand{\Coef}{\mathrm{Coef}}

\DeclareMathOperator{\Aut}{Aut}
\newcommand{\Ch}{\mathrm{Ch}}

\newcommand{\tQ}{\widetilde{Q}}

\newcommand{\of}{\overline{f}}

\newcommand{\tK}{\widetilde{K}}

\newcommand{\triv}{\mathrm{triv}}

\newcommand{\tP}{\widetilde{P}}

\newcommand{\tf}{\widetilde{f}}

\newcommand{\od}{\overline{d}}

\newcommand{\mbE}{\mathbb{E}}

\newcommand{\e}{\mathrm{e}}

\newcommand{\lb}{\left(}
\newcommand{\rb}{\right)}

\newcommand{\R}{\mathrm{R}}
\newcommand{\Mi}{\mathrm{Mi}}
\newcommand{\norm}{\mathrm{norm}}
\newcommand{\tT}{\widetilde{T}}
\newcommand{\Ad}{\mathrm{Ad}}
\newcommand{\tlambda}{\widetilde{\lambda}}
\newcommand{\cDO}{\mathcal{DO}}
\newcommand{\hcDO}{\widehat{\mathcal{DO}}}
\newcommand{\tF}{\widetilde{F}}

\newtheorem{theorem}{Theorem}[section]
\newtheorem{proposition}[theorem]{Proposition}
\newtheorem{lemma}[theorem]{Lemma}
\newtheorem{corollary}[theorem]{Corollary}
\newtheorem{conjecture}[theorem]{Conjecture}

\theoremstyle{definition}

\newtheorem{remark}[theorem]{Remark}

\newtheorem{definition}[theorem]{Definition}

\usepackage{color}

\numberwithin{equation}{section}

\renewcommand{\gg}[2]{\fill[color=white] (#2) circle(2.5mm) node {\color{black}$\substack{#1}$}; \draw (#2) circle (2.5mm)}
\newcommand{\lab}[4]{\draw (#1)++(#2:#3) node {$\substack{#4}$};}
\newcommand{\leg}[2]{\begin{scope}[shift={(#1)}] \draw (0:0) -- (#2:6.1mm);\end{scope}}
\newcommand{\legm}[3]{\begin{scope}[shift={(#1)}] \draw (0:0) -- (#2:7.8mm);\fill[color=white] (#2:7.8mm) circle(1.7mm) node {\color{black}$\substack{#3}$};\end{scope}}

\begin{document}

\title[Deformations of the Riemann hierarchy and the geometry of $\oM_{g,n}$]{Deformations of the Riemann hierarchy and the geometry of $\oM_{g,n}$}

\author{Alexandr Buryak}
\address{A. Buryak:\newline 
Faculty of Mathematics, National Research University Higher School of Economics, Usacheva str. 6, Moscow, 119048, Russian Federation;\smallskip\newline 
Skolkovo Institute of Science and Technology, Bolshoy Boulevard 30, bld. 1, Moscow, 121205, Russian Federation}
\email{aburyak@hse.ru}

\author{Paolo Rossi}
\address{P. Rossi:\newline Dipartimento di Matematica ``Tullio Levi-Civita'', Universit\`a degli Studi di Padova,\newline
Via Trieste 63, 35121 Padova, Italy}
\email{paolo.rossi@math.unipd.it}

\begin{abstract}
The Riemann hierarchy is the simplest example of rank one, ($1$+$1$)-dimensional integrable system of nonlinear evolutionary PDEs. It corresponds to the dispersionless limit of the Korteweg--de Vries hierarchy. In the language of formal variational calculus, we address the classification problem for deformations of the Riemann hierarchy satisfying different extra requirements (general deformations, deformations as systems of conservation laws, Hamiltonian deformations, and tau-symmetric deformations), under the natural group of coordinate transformations preserving each of those requirements. We present several results linking previous conjectures of Dubrovin--Liu--Yang--Zhang (for the tau-symmetric case) and of Arsie--Lorenzoni--Moro (for systems of conservation laws) to the double ramification hierarchy construction of integrable hierarchies from partial CohFTs and F-CohFTs. We prove that, if the conjectures are true, DR hierarchies of rank one are universal objects in the space of deformations of the Riemann hierarchy. We also prove a weaker version of the DLYZ conjecture and that the ALM conjecture implies (the main part of) the DLYZ conjecture. Finally we characterize those rank one F-CohFTs which give rise to Hamiltonian deformations of the Riemann hierarchy.
\end{abstract}

\date{\today}

\maketitle

\tableofcontents

\section{Introduction}

This paper deals with deformations of the Riemann hierarchy, the simplest integrable hierarchy of nonlinear evolutionary PDEs in one function $u=u(x, t_0, t_1, t_2, \ldots)$ of one space and one time variable, namely,
$$
\frac{\d u}{\d t_d}=\frac{u^d}{d!}\frac{\d u}{\d x},\quad d\in\mbZ_{\ge 0}.
$$
The deformations we consider have the form $\frac{\d u}{\d t_d} = Q_d(u,u_x,u_{xx},\ldots)$, with $Q_d$ a power series in $u$ and its $x$-derivatives $u_x = \frac{\d u}{\d x}, u_{xx} = \frac{\d^2 u}{\d x^2}, \ldots$ such that, rescaling $t_d \mapsto t_d/\eps$ and $x\mapsto x/\eps$ the equations become $\frac{\d u}{\d t_d}=\frac{u^d}{d!}u_x + O(\eps)$.

\medskip

The Riemann hierarchy has several interesting properties: it is a system of conservation laws, Hamiltonian and tau-symmetric. This last property means that it is possible to associate to any solution a single function of $x$ and the time variables $t_0,t_1,\ldots$ (the tau-function) encoding the evolution of all of its Hamiltonian densities. In view of their classification, then, one can decide to enforce some, or all, of the mentioned properties on the entire deformation.

\medskip

For instance, in \cite{DLYZ16}, the authors focused on tau-symmetric deformations, relating them to hierarchies of topological type for the Hodge classes on the moduli spaces of stable algebraic curves, i.e. integrable hierarchies whose tau-function for the topological solution with initial datum $u(x, t_0=0, t_1=0, \ldots) = x$ coincides with the generating series of intersection numbers of rank one cohomological field theories (CohFTs) with psi-classes. Their main conjecture is that these Hodge hierarchies cover all possible tau-symmetric deformations of the Riemann hierarchy. In particular they suggest a specific \emph{standard form} for each tau-symmetric deformation, i.e. a family of deformations such that each element in the family is a represenatative of a different equivalence class, under coordinate change, of tau-symmetric deformations and all classes have a representative in the family. This family is an explicit deformation depending on complex parameters and part of the aforementioned DLYZ conjecture identifies which of these parameters are independent, so that they effectively serve as coodinates for the space of inequivalent deformations.

\medskip

A similar classification conjecture, based on explorations with different techniques, can be found in \cite{ALM15b} in the case of general deformations. In particular, for deformations that are conservation laws (i.e. the differential polynomial on the right-hand side of each evolutionary PDE in the hierarchy is a total $x$-derivative) the authors also give a \emph{normal form} for each deformation of conservation law type. In this case functional parameters appear to control the deformation and part of the ALM conjecture identifies which of these functional parameters are independent.

\medskip

In this paper we study the interaction of these conjectures with the family of rank one integrable systems constructed from the intersection theory of the moduli space of stable algebraic curves using the double ramification (DR) hierarchy construction introduced in \cite{Bur15-DR-hierarchy}. Recall that the DR hierarchy associates to a partial cohomological field theory (partial CohFT), i.e. a family of cohomology classes on each moduli space of curves with given genus and number of marked point, compatible with the natural morphisms between such moduli spaces, a tau-symmetric integrable system of Hamiltonian PDEs \cite{BDGR18}. The construction was generalized in \cite{ABLR21} to F-cohomological field theories (F-CohFTs), satisfying weaker axioms and giving rise more generally to integrable systems of conservation laws.

\medskip

The main idea behind this work is that \emph{DR hierarchies for rank one CohFTs or F-CohFTs should be universal objects for the corresponding classes of deformations of the Riemann hierarchy}.

\medskip

Indeed, in Theorem \ref{theorem:DR for rank one partial CohFTs} we prove that, for a certain explicit family of partial CohFTs, the corresponding DR hierarchies are automatically in the DLYZ standard form, and that the family as enough parameters to span all values of what, conjecturally, should be the whole set of independent parameters appearing in the DLYZ conjecture. In other words, if the DLYZ conjecture is true, then Theorem \ref{theorem:DR for rank one partial CohFTs} implies that the DR hierarchies for partial CohFTs give all tau-symmetric deformations of the Riemann hierarchy, directly in the standard form.

\medskip

Similarly, in Theorem \ref{theorem:DR hierarchy for F-CohFT} we prove that, for a certain explicit family of F-CohFTs, the corresponding DR hierarchies are automatically in the ALM normal form, and that the family as enough parameters to span the subset of constant values of what, conjecturally, should be the whole set of independent functional parameters appearing in the ALM conjecture. In other words, if the ALM conjecture is true, then  Theorem \ref{theorem:DR hierarchy for F-CohFT} implies that the DR hierarchies for F-CohFTs give all ALM normal forms of conservation laws deforming the Riemann hierarchy whose functional parameters are constant.

\medskip

In fact before doing that, we also make partial progress in the proof of the DLYZ conjecture itself, by proving in Theorem~\ref{theorem:generalized standard form} a somewhat weaker version of their standard form conjecture. Our \emph{generalized standard form} potentially allows for extra terms compared to the standard form of DLYZ and the DLYZ conjecture is reduced to the automatic vanishing of these extra terms. Moreover, we prove that the ALM conjecture implies (the main part of) the DLYZ conjecture.

\medskip

The structure of the paper is the following. In Section \ref{section:preliminaries} we prove several results, some of them of independent interest, on the natural groups of coordinate transformations (subgroups of the Miura group) preserving the Hamiltonian structure or the tau-symmetry of a system of evolutionary PDEs, and on their interaction/intersection. Moreover we briefly recall the DR hierarchy construction for partial CohFTs and F-CohFTs.

\medskip

In Section \ref{section:tau-symmetric deformations} we recall the DLYZ conjecture and we prove a weaker version of this conjecture, on the generalized standard form of tau-symmetric deformations of the Riemann hierarchy. Moreover, we show that the DR hierarchies of rank one partial CohFTs are in DLYZ standard form and, for an explicit family of partial CohFTs, span the whole set of inequivalent deformations.

\medskip

In Section \ref{section:conservation laws} we recall the ALM conjecture, we show that the DR hierarchies of rank one F-CohFTs are in the ALM normal form and, for an explicit family of F-CohFTs, span the set of inequivalent deformations corresponding to constant functional parameters. We then prove that the ALM conjecture implies (the main part of) the DLYZ conjecture. Finally, among the F-CohFTs in the aforementioned explicit family, we characterize those for which the corresponding DR hierarchy is a Hamiltonian deformation of the Riemann hierarchy.

\medskip

\noindent{\bf Notation and conventions.}
\begin{itemize}
\item For a topological space $X$, we denote by $H^i(X)$ the cohomology groups with coefficients in $\mbC$. Let $H^{\e}(X):=\bigoplus_{i\ge 0}H^{2i}(X)$. For a cohomology class $\omega\in H^*(X)$, we will denote by $[\omega]_d$ the image of $\omega$ under the projection $H^*(X)\to H^{2d}(X)$. 

\smallskip

\item We will work with the moduli spaces $\oM_{g,n}$ of stable algebraic curves of genus $g$ with $n$ marked points, which are nonempty only when the condition $2g-2+n>0$ is satisfied. We will often omit mentioning this condition explicitly, and silently assume that it is satisfied when a moduli space is considered.

\smallskip

\item For $n\in\mbZ_{\ge 0}$, we denote $[n]:=\{1,\ldots,n\}$.

\smallskip

\item Let us fix some notations regarding partitions.
\begin{itemize}
\item We denote by $\cP_n$ the set of partitions of $n\in\mbZ_{\ge 0}$. For $\lambda=(\lambda_1,\ldots,\lambda_l)\in\cP_n$, let 
$$
m_k(\lambda):=|\{1\le i\le l|\lambda_i=k\}|,\qquad (\lambda,1):=(\lambda_1,\ldots,\lambda_l,1)\in\cP_{n+1}.
$$

\item We will consider the following subsets of $\cP_n$:
$$
\cP_n^\circ:=\left\{\lambda\in\cP_n|l(\lambda)\ge 2,\lambda_1=\lambda_2\right\},\qquad
\cP_n':=\left\{\lambda\in\cP_n^\circ|\lambda_i\ge 2\right\}.
$$

\item We will consider the lexicographical order on partitions from $\cP_n$, where $\lambda<\mu$ if and only if $\lambda_1=\mu_1,\ldots,\lambda_k=\mu_k$, and $\lambda_{k+1}<\mu_{k+1}$, for some $k$.

\end{itemize}
\end{itemize}

\medskip

\noindent{\bf Acknowledgements.} The work of A.~B. is an output of a research project implemented as part of the Basic Research Program at the National Research University Higher School of Economics (HSE University). P.~R. is supported by the University of Padova and is affiliated to the INFN under the national project MMNLP and to the INdAM group GNSAGA.

\medskip


\section{Preliminaries}\label{section:preliminaries}

\subsection{Evolutionary PDEs}

We will consider only evolutionary PDEs with one dependent variable. We will use the following notations (see e.g. the reviews~\cite{Ros17,Bur24} for more details):
\begin{itemize}
\item $\cA_u:=\mbC[[u]][u_x,u_{xx},\ldots]$ and $\hcA_u:=\cA_u[[\eps]]$ are the algebras of \emph{differential polynomials}. By $\hcA_{u;d}\subset\hcA_u$, we denote the subspace of differential polynomials of (differential) degree~$d$. For $P\in\hcA_u$, we will denote by $P^{\<d\>}$ the image of $P$ under the projection $\hcA_u\to\bigoplus_{i=0}^d\eps^i\cA_u\subset\hcA_u$.

\smallskip

\item For any $P\in\hcA_u$, define an operator $D_P:=\sum_{n\ge 0}(\d_x^n P)\frac{\d}{\d u_n}$ on $\hcA_u$. Such an operator is called an \emph{evolutionary operator}.

\smallskip

\item $\Lambda_u:=\cA_u/(\Im(\d_x)\oplus\mbC)$ and $\hLambda_u:=\Lambda_u[[\eps]]$ are the vector spaces of \emph{local functionals}. The \emph{variational derivative} $\frac{\delta}{\delta u}:=\sum_{n\ge 0}(-\d_x)^n\circ\frac{\d}{\d u_n}$ is correctly defined on local functionals, $\frac{\delta}{\delta u}\colon\hLambda_u\to\hcA_u$, and a local functional $\oh\in \hLambda_u$ is equal to zero if and only if $\frac{\delta\oh}{\delta u}=0$ (see~e.g. \cite[Lemma~2.1.6]{LZ11}).

\smallskip

\item We denote by $\cDO_u$ the vector space of differential operators $K$ of the form $K=\sum_{i}K_i\d_x^i$, $K_i\in\cA_u$, where the sum is finite. We will use the notation $K^\dagger:=\sum_i (-\d_x)^i\circ K_i$. Denote $\hcDO_u:=\cDO_u[[\eps]]$. For any $P\in\hcA_u$, let us denote $L(P):=\sum_{n\ge 0}\frac{\d P}{\d u_n}\d_x^n\in\hcDO_u$. We will use the fact that $P\in\Im(\frac{\delta}{\delta u})$ if and only if $L(P)=L(P)^\dagger$~\cite{Dor78}.

\smallskip

\item Any differential operator $K\in\hcDO_u$ defines a bracket $\{\cdot,\cdot\}_K$ on the space $\hLambda_u$ by $\{\of,\oh\}_K:=\int\frac{\delta\of}{\delta u}K\frac{\delta\oh}{\delta u}dx$. If this bracket is skewsymmetric and satisfies the Jacobi identity, then the operator $K$ is called \emph{Poisson}.

\smallskip

\item We will consider systems of PDEs
\begin{gather}\label{eq:general system}
\frac{\d u}{\d t_i}=P_i,\quad P_i\in\hcA_{u;1},\quad i\in\mbZ_{\ge 1},
\end{gather}
and we will say that the flows pairwise commute if $[D_{P_i},D_{P_j}]=0$ for all $i,j\ge 1$. One says that the system~\eqref{eq:general system} is
\begin{itemize}
\item a \emph{system of conservation laws} if $P_i\in\Im(\d_x)$;
\item \emph{Hamiltonian} if $P_i=K\frac{\delta\oh_i}{\delta u}$ for some Poisson operator $K\in\hcDO_u$ and local functionals $\oh_i\in\hLambda_u$.
\end{itemize}
\end{itemize}

\medskip

We will often use the following well-known lemma~(see \cite[Lemma~3.3]{LZ06} and \cite[Lemma 2.5]{Bur15}).

\begin{lemma}\label{lemma:uniqueness}
Consider a differential polynomial $P\in\hcA_{u;1}$ of the form $P=uu_x+O(\eps)$ and an integer $d\ge 1$.
\begin{enumerate}
\item If a differential polynomial $Q\in\hcA_{u;1}$ satisfies $[D_P,D_Q]=O(\eps^{d+1})$, then $Q^{\<d\>}$ is uniquely determined by $P^{\<d\>}$ and $Q|_{\eps=0}$.

\smallskip

\item If a local functional $\oh\in\hLambda_{u;0}$ satisfies $D_P(\oh)=O(\eps^{d+1})$, then $\oh^{\<d\>}$ is uniquely determined by $P^{\<d\>}$ and $\oh|_{\eps=0}$.
\end{enumerate}
\end{lemma}

\medskip

For a partition $\lambda$, let $u_\lambda:=\prod_{i=1}^{l(\lambda)}u_{\lambda_i}$. We will need several times the following technical statement, which is similar to~\cite[Lemma~3.2]{LZ06}.

\begin{lemma}\label{lemma:differentiating ulambda}
For an arbitrary partition $\lambda$, we have
$$
D_{uu_x}(u_\lambda)=(|\lambda|+l(\lambda)-m_1(\lambda)-1)u_\lambda u_x+\d_x(uu_\lambda)+\sum_{\substack{\mu\in\cP_{|\lambda|+1}\\\mu<(\lambda,1)}}d_{\lambda,\mu}u_\mu,
$$
where $d_{\lambda,\mu}$ are some integer coefficients.
\end{lemma}
\begin{proof}
Straightforward computation.
\end{proof}

\medskip

Recall the following statement (see e.g. Lemma~8.2 from~\cite{BDGR20} and its proof).

\begin{lemma}\label{lemma:unique density}{\ }
\begin{enumerate}
\item For any local functional $\oh\in\hLambda_{u;0}$, there exists a unique density $h\in\hcA_{u;0}$, $\oh=\int h dx$, such that
$$
h=f+\sum_{k\ge 2}\eps^k\sum_{\lambda\in\cP_k^\circ}f_\lambda u_\lambda,\quad f,f_\lambda\in\mbC[[u]],\quad f(0)=0.
$$

\smallskip

\item For any $n\ge 2$ and $\lambda\in\cP_n$, there exists a unique collection of rational constants $e_{\lambda,\mu}$, $\mu\in\cP_n^\circ$, $\mu\le\lambda$, such that
$$
\int u_\lambda dx=\int\lb\sum_{\mu\in\cP_n^\circ,\,\mu\le\lambda}e_{\lambda,\mu}u_\mu\rb dx.
$$
\end{enumerate}
\end{lemma}

\medskip

The lemma immediately implies the following corollary.

\begin{corollary}\label{corollary:simple}
Consider a local functional $\oh\in\hLambda_{u;0}$.
\begin{enumerate}
\item If $\frac{\d\oh}{\d u}=0$, then $\oh=\alpha\int u dx+\int h dx$ for some $\alpha\in\mbC$ and $h\in\hcA_{u;0}$ satisfying $\frac{\d h}{\d u}=0$. 

\smallskip

\item If $\frac{\delta\oh}{\delta u}\in\Im(\d_x)$, then $\oh=\int h dx$ for some $h\in\hcA_{u;0}$ satisfying $\frac{\d h}{\d u}=0$.
\end{enumerate}
\end{corollary}

\medskip

\begin{definition}{\ }
\begin{enumerate}
\item By a \emph{Miura transformation}, we mean a change of variables of the form $u\mapsto\tu(u_*,\eps)=u+\eps f$, $f\in\hcA_{u;1}$. We denote the group of Miura transformations by~$\Mi$. Let us denote~by 
$$
\Mi^{(k)}\subset\Mi,\quad k\ge 1,
$$
the subgroup consisting of Miura transformations of the form $u\mapsto\tu(u_*,\eps)=u+\eps^k f$, $f\in\hcA_{u;k}$.

\smallskip

\item An \emph{elementary} Miura transformation is a Miura transformation of the form $u\mapsto\tu(u_*,\eps)=u+\eps^k f$, where $k\ge 1$ and $f\in\cA_{u;k}$. Let us denote this elementary Miura transformation by $\Phi_f$. The group $\Mi$ is generated by elementary Miura transformations, i.e. any Miura transformation can be expressed as the composition $\cdots\circ\Phi_{f_3}\circ\Phi_{f_2}\circ\Phi_{f_1}$ for some $f_k\in\cA_{u;k}$.

\smallskip

\item In the group of Miura transformations, consider the subgroup that preserves the operator~$\d_x$. Let us denote it by 
$$
\Mi_{\d_x}\subset\Mi.
$$
We will also use the notation $\Mi_{\d_x}^{(k)}:=\Mi_{\d_x}\cap\Mi^{(k)}$.

\smallskip

\item A \emph{normal Miura transformation} is a Miura transformation of the form:
\begin{gather}\label{eq:normal Miura}
u\mapsto\tu(u_*,\eps)=u+\d_x^2P,\quad P\in\hcA_{u;-2}.
\end{gather}
Denote the group of normal Miura transformations by $\Mi_\norm$. 
\end{enumerate}
\end{definition}

\medskip

Let us recall the following facts about the group $\Mi_{\d_x}$ (see e.g.~\cite[Theorem~2.3]{Dub10}):
\begin{itemize}
\item The group $\Mi_{\d_x}$ is generated by Miura transformations of the form
\begin{gather}\label{eq:Miura-Phihdx}
u\mapsto\tu(u_*,\eps)=u+\eps^k\{u,\oh\}_{\d_x}+\frac{\eps^{2k}}{2!}\{\{u,\oh\}_{\d_x},\oh\}_{\d_x}+\frac{\eps^{3k}}{3!}\{\{\{u,\oh\}_{\d_x},\oh\}_{\d_x},\oh\}_{\d_x}+\ldots,
\end{gather}
where $k\ge 1$ and $\oh\in\Lambda_{u;k-1}$. Let us denote this Miura transformation by $\Phi_{\oh,\d_x}$. 

\smallskip

\item For any Miura transformation of the form~\eqref{eq:Miura-Phihdx} and $\of\in\hLambda_u$, we have
$$
\of[\tu]=\left.\lb\of+\sum_{i\ge 1}\eps^{ki}\Ad_{\oh,\d_x}^i(\of)\rb\right|_{u_j\mapsto\tu_j},
$$
where $\Ad_{\oh,\d_x}(\of):=\{\oh,\of\}_{\d_x}$.

\smallskip

\item The previous fact immediately implies that for any Miura transformation $u\mapsto\tu(u_*,\eps)$ from the group $\Mi_{\d_x}$ we have $\lb\int\frac{u^2}{2}dx\rb[\tu]=\int\frac{\tu^2}{2}dx$.

\end{itemize}

\medskip

\begin{lemma}
The group $\Mi_{\d_x}\cap\Mi_\norm$ is generated by the Miura transformations $\Phi_{\oh,\d_x}$, $\oh\in\Lambda_{u;k-1}$, $k\ge 2$, with $\frac{\d\oh}{\d u}=0$.
\end{lemma}
\begin{proof}
The fact that these Miura transformations are normal is obvious, because $\d_x\frac{\delta\oh}{\delta u}\in\Im(\d_x^2)$. Consider a Miura transformation from the group $\Mi_{\d_x}\cap\Mi_\norm$:
$$
u\mapsto\tu(u_*,\eps)=u+\eps^k\d_x^2 P+O(\eps^{k+1}),\quad k\ge 2,\quad P\in\cA_{u;k-2}.
$$
We compute 
$$
L(\tu(u_*,\eps))\circ\d_x\circ L(\tu(u_*,\eps))^\dagger=\d_x+\eps^k\d_x\circ\lb L(\d_x P)-L(\d_x P)^\dagger\rb\circ\d_x+O(\eps^{k+1}).
$$
Since our Miura transformation belongs to the group $\Mi_{\d_x}$, we conclude that $L(\d_x P)-L(\d_x P)^\dagger=0$ and therefore $\d_x P=\frac{\delta\oh_1}{\delta u}$ for some $\oh_1\in\Lambda_{u;k-1}$. By Corollary~\ref{corollary:simple}, $\frac{\d\oh_1}{\d u}=0$. Taking the composition of our Miura transformation with $\Phi_{-\oh_1,\d_x}$, we obtain a Miura transformation from the group $\Mi^{(k+1)}_{\d_x}\cap\Mi_\norm$. Continuing this procedure, we express our Miura transformation as the composition $\cdots\circ\Phi_{\oh_3.\d_x}\circ\Phi_{\oh_2,\d_x}\circ\Phi_{\oh_1,\d_x}$ with $\oh_i\in\Lambda_{u;k-2+i}$ satisfying $\frac{\d\oh_i}{\d u}=0$. 
\end{proof}

\medskip

It is well known that any Poisson operator~$K$ of degree $1$ satisfying $K|_{\eps=0}=
\d_x$ can be transformed to~$\d_x$ by some Miura transformation~\cite{Get02,DMS05}. The next proposition describes when such a Poisson operator can be transformed to $\d_x$ by a normal Miura transformation.

\medskip

\begin{proposition}\label{proposition:K to dx by normal}
A Poisson operator~$K$ of degree $1$ satisfying $K|_{\eps=0}=\d_x$ can be transformed to $\d_x$ by a normal Miura transformation if and only if $K=\tK\circ\d_x$ for some operator $\tK$. 
\end{proposition}
\begin{proof}
For the ``only if'' part, we consider an arbitrary normal Miura transformation $u\mapsto\tu(u_*,\eps)=u+\d_x^2 P$ and compute $L(\tu(u_*,\eps))\circ\d_x\circ L(\tu(u_*,\eps))^\dagger=L(\tu(u_*,\eps))\circ\lb 1+\d_x\circ L(P)^\dagger
\circ \d_x\rb\circ\d_x$, which proves the claim.

\medskip

To prove the ``if'' part, consider a Poisson operator $K$ of degree $1$ satisfying $K|_{\eps=0}=\d_x$ and having the form $K=\tK\circ\d_x$. We will construct a required normal Miura transformation as the composition of elementary normal Miura transformations in the following way. Suppose that~$K$ has the form $K=\d_x+\eps^i K^{[i]}+O(\eps^{i+1})$, $i\ge 1$. We know that there exists an elementary Miura transformation $u\mapsto \tu(u_*,\eps)=u+\eps^i f_i$, $f_i\in\cA_{u;i}$, satisfying $K_\tu=\d_x+O(\eps^{i+1})$. On the other hand, 
$$
\left.K_{\tu}\right|_{\tu_n\mapsto u_n}=\d_x+\eps^i\lb\underline{L(f_i)\circ\d_x+\d_x\circ L(f_i)^\dagger+K^{[i]}}\rb+O(\eps^{i+1}).
$$
The coefficient of $\d_x^0$ in the underlined operator is equal to $\d_x\frac{\delta f_i}{\delta u}$. Since it is zero, we obtain $f_i=\d_x\tf_i$ for some $\tf_i\in\cA_{u;i-1}$. Then we have
$$
\left.K_{\tu}\right|_{\tu_n\mapsto u_n}=\d_x+\eps^i\lb\d_x\circ\lb L(\tf_i)- L(\tf_i)^\dagger\rb\circ\d_x+K^{[i]}\rb+O(\eps^{i+1}).
$$
Let us find a differential polynomial $r_i\in\cA_{u;i-1}$ such that $\tf_i=\frac{\d r_i}{\d u}$. Clearly we have
$$
\tf_i=\frac{\delta r}{\delta u}+\d_x\Big(\underbrace{\sum_{s\ge 1}(-\d_x)^{s-1}\frac{\d r_i}{\d u_s}}_{f_i':=}\Big)\quad\text{and}\quad L(\tf_i)- L(\tf_i)^\dagger=L(\d_x f_i')- L(\d_x f_i')^\dagger,
$$
which implies that the elementary normal Miura transformation $u\mapsto u'(u_*,\eps)=u+\eps^i\d_x^2 f_i'$ satisfies $K_{u'}=\d_x+O(\eps^{i+1})$. We can apply the same procedure for the operator $K_{u'}$ and so on. The required normal Miura transformation is then the composition of the constructed elementary normal Miura transformations.
\end{proof}

\medskip

\subsection{Partial CohFTs, F-CohFTs, and the DR hierarchies} \label{section:CohFTs}

Cohomological field theories (CohFTs), partial CohFTs, and F-CohFTs were introduced in \cite{KM94}, \cite{LRZ15}, and \cite{BR18}, respectively. For the general definitions, and a discussion on their differences, the reader is referred for instance to \cite[Section 3.1]{ABLR23}. In this paper we will only consider partial CohFTs and F-CohFTs of rank $1$, with phase space $V=\mbC$ and unit $e=1$. In the case of partial CohFTs, the metric $\eta$ will be always given by $\eta(e,e)=1$.

\medskip

So, for this paper, a partial CohFT is a family of cohomology classes 
$$
c_{g,n}\in H^\e(\oM_{g,n}),\quad 2g-2+n>0,
$$ 
invariant under permutation of the $n$ marked points and satisfying $c_{0,3}=1$, $\pi^*c_{g,n} = c_{g,n+1}$ for the map $\pi\colon\oM_{g,n+1} \to \oM_{g,n}$ forgetting the last marked point, and $\sigma^* c_{g_1+g_2,n_1+n_2}= c_{g_1,n_1+1} \times c_{g_2,n_2+1}$ for the map $\sigma\colon \oM_{g_1,n_1+1}\times \oM_{g_2,n_2+1} \to \oM_{g_1+g_2,n_1+n_2}$ joining two stable curves at their last marked points to form a nodal curve.

\medskip

The definition of an F-CohFT differs from the one of a partial CohFT in the additional requirement $n\ge 1$, while invariance is under permutation of the last $n-1$ marked points only. The rest of the axioms are the same, which means in particular that restricting a partial CohFT to moduli spaces with $n\ge 1$ gives an F-CohFT.

\medskip

For any $G\in\mbC$, consider the following partial CohFT:
$$
c_{g,n}^{\triv,G}:=G^g\in H^0(\oM_{g,n}).
$$
The corresponding F-CohFT will be denoted by the same symbol.

\medskip

For any formal power series $R(z)\in 1+z\mbC[[z]]$ satisfying $R(z)R(-z)=1$, there is a partial CohFT denoted by $\{R.c^{\triv,G}_{g,n}\}$ and defined by the standard Givental formula for CohFTs, as described for instance in \cite{PPZ15}, but where the sum runs over stable trees only, instead of the usual general stable graphs. In~\cite[Section~4]{ABLR23}, the authors defined an F-CohFT $\{R.c^{\triv,G}_{g,n+1}\}$ for an arbitrary $R(z)\in 1+z\mbC[[z]]$  by a similar formula where the sum runs again over stable trees only.

\medskip

For any F-CohFT $\{c_{g,n+1}\}$, we have the associated \emph{DR hierarchy}
$$
\frac{\d u}{\d t_d}=\d_x P_d,\quad d\ge 0,
$$
where 
$$
P_d:=\sum_{g,n\ge 0}\frac{\eps^{2g}}{n!}\sum_{\substack{d_1,\ldots,d_n\ge 0\\d_1+\ldots+d_n=2g}}\Coef_{a_1^{d_1}\cdots a_n^{d_n}}\lb\int_{\oM_{g,n+2}}\hspace{-0.7cm}\psi_2^d\lambda_g\DR_g\lb-\sum a_i,0,a_1,\ldots,a_n\rb c_{g,n+2}\rb u_{d_1}\cdots u_{d_n}.
$$
Here $\psi_i\in H^2(\oM_{g,n},\mbQ)$, for $1\leq i\leq n$, denotes the first Chern class of the $i$-th tautological line bundle on $\oM_{g,n}$ whose fiber over a curve is the cotangent line to the curve at the $i$-th marked point, $\lambda_g \in H^{2g}(\oM_{g,n},\mbQ)$ is the top Chern class of the rank $g$ Hodge bundle~$\mbE$ on~$\oM_{g,n}$ whose fiber over a curve is the $g$-dimensional vector space of holomorphic differentials on the curve, and $\DR_g(a_1,\ldots,a_n) \in H^{2g}(\oM_{g,n},\mbQ)$ is the double ramification cycle, a cohomology class, polynomial of degree $2g$ in the variables $a_1,\ldots,a_n$, which represents a compactification by relative stable maps to $\mbP^1$ of the locus of smooth curves whose marked points form the support of a principal divisor with multiplicities $a_1,\ldots,a_n$. More details on these classes and the DR hierarchy construction can be found in the original paper \cite{Bur15-DR-hierarchy} introducing the DR hierarchy and its F-CohFT counterpart \cite{ABLR21}.
 
\medskip

It is straightforward from the above definition in terms of intersection numbers that the differential polynomials $P_d$ have the form $P_d=\frac{u^{d+1}}{(d+1)!}+O(\eps)$. By~\cite[Theorem~3]{ABLR21}, the DR hierarchy is endowed with conserved quantities $\og_d=\int g_d dx$, $d\ge 0$, given~by
$$
g_d:=\sum_{g,n\ge 0}\frac{\eps^{2g}}{n!}\sum_{\substack{d_1,\ldots,d_n\ge 0\\d_1+\ldots+d_n=2g}}\Coef_{a_1^{d_1}\cdots a_n^{d_n}}\lb\int_{\oM_{g,n+1}}\hspace{-0.5cm}\psi_1^d\lambda_g\DR_g\lb-\sum a_i,a_1,\ldots,a_n\rb c_{g,n+1}\rb u_{d_1}\cdots u_{d_n}.
$$
So we have $D_{\d_xP_{d_1}}\lb\og_{d_2}\rb=0$ for any $d_1,d_2\ge 0$. The conserved quantities $\og_d$ have the form $\og_d=\int\lb\frac{u^{d+2}}{(d+2)!}+O(\eps)\rb dx$.

\medskip

If our F-CohFT is associated to a partial CohFT, then the DR hierarchy is Hamiltonian with Poisson operator $\d_x$ and Hamiltonians $\og_d$, i.e., $P_d=\frac{\delta\og_d}{\delta u}$. Moreover it is tau-symmetric with tau-structure $h_{d-1}=\frac{\delta \og_{d}}{\delta u}$, where $\og_{d}=\oh_d$, for $d\geq 0$, as shown in \cite{BDGR18}.

\medskip


\section{Tau-symmetric deformations of the Riemann hierarchy and partial CohFTs}\label{section:tau-symmetric deformations}

The \emph{Riemann hierarchy} is the following system of pairwise commuting flows:
$$
\frac{\d u}{\d t_d}=\frac{u^d}{d!}u_x,\quad d\in\mbZ_{\ge 0}.
$$
It is Hamiltonian with the Hamiltonians $\oh^\R_d:=\int\frac{u^{d+2}}{(d+2)!}dx$ and the Poisson operator $\d_x$, $\frac{u^d}{d!}u_x=\d_x\frac{\delta\oh_d^\R}{\delta u}$.

\medskip

\subsection{Deformations of the Riemann hierarchy}

\begin{definition}{\ }
\begin{enumerate}
\item A \emph{deformation} of the Riemann hierarchy is a sequence of differential polynomials $Q_d\in\hcA_{u;1}$, $d\ge 0$, such that
\begin{itemize}
\item $Q_d|_{\eps=0}=\frac{u^d}{d!}u_x$.
\item The flows $\frac{\d u}{\d t_d}=Q_d$, $d\ge 0$, pairwise commute.
\end{itemize}

\smallskip

\item A deformation of the Riemann hierarchy is called \emph{even} if $Q_d\in\cA_u[[\eps^2]]\subset\hcA_u$, $d\ge 0$.
\end{enumerate}
\end{definition}

\medskip

By Lemma~\ref{lemma:uniqueness}, a deformation of the Riemann hierarchy is uniquely determined by the differential polynomial $Q_1$. Note that for any deformation we have $Q_0=u_x$.

\medskip

\begin{definition}{\ }
\begin{enumerate}
\item A deformation of the Riemann hierarchy is called \emph{Hamiltonian} if it has the form
$$
\frac{\d u}{\d t_d}=K\frac{\delta\oh_d}{\delta u},\quad d\ge 0,
$$
where $\oh_d=\oh_d^\R+O(\eps)\in\hLambda_{u;0}$ and $K=\d_x+O(\eps)$ is a Poisson operator of degree~$1$.

\smallskip

\item A Hamiltonian deformation is called \emph{special} if $K=\d_x$ and $\frac{\d\oh_1}{\d u}=\oh_0$. Note that since $Q_1=u_x$ and $K=\d_x$ we have $\oh_0=\int\frac{u^2}{2}dx$.

\smallskip

\item A \emph{tau-symmetric deformation} of the Riemann hierarchy is a Hamiltonian deformation together with a choice of densities $h_d=\frac{u^{d+2}}{(d+2)!}+O(\eps)\in\hcA_{u;0}$ for the Hamiltonians $\oh_d$ satisfying the following properties:
\begin{itemize}
\item $\{h_{p-1},\oh_q\}_K=\{h_{q-1},\oh_p\}_K$, $p,q\ge 0$, where $h_{-1}:=u$.

\item $\{u,\oh_0\}_K=u_x$ and $\{u,\oh_{-1}\}_K=0$.
\end{itemize}
The collection of differential polynomials $h_d$ is called the \emph{tau-structure}.
\end{enumerate}
\end{definition}

\medskip

\begin{remark}\label{remark:about deformations}{\ }
\begin{enumerate}
\item For any special Hamiltonian deformation of the Riemann hierarchy, we have $\frac{\d\oh_d}{\d u}=\oh_{d-1}$ for all $d\ge 1$. This is proved by writing 
$$
0=\frac{\d}{\d u}\{\oh_1,\oh_d\}_{\d_x}-\{\oh_1,\oh_{d-1}\}_{\d_x}=\left\{\oh_1,\frac{\d}{\d u}\oh_d-\oh_{d-1}\right\}_{\d_x},
$$
which, since $\frac{\d}{\d u}\oh_d-\oh_{d-1}=O(\eps)$, implies that $\frac{\d}{\d u}\oh_d-\oh_{d-1}=0$. 

\smallskip

\item By \cite[Lemma~3.3]{BDGR18}, if for a given Hamiltonian deformation of the Riemann hierarchy a tau-structure exists, then it is unique and moreover is given by (see equation~(3.17) in \cite{BDGR18})
$$
\d_x h_{p-1}=K\frac{\delta\oh_p}{\delta u},\quad p\ge 0.
$$

\smallskip

\item The previous equality implies that if $K=\d_x$, then $\oh_{p-1}=\frac{\d\oh_p}{\d u}$, and therefore the Hamiltonian deformation is special. Conversely, by \cite[Proposition~3.1]{BDGR18}, an arbitrary special Hamiltonian deformation of the Riemann hierarchy is tau-symmetric: a tau-structure is given by $h_{p-1}=\frac{\delta\oh_p}{\delta u}$, $p\ge 0$. Thus, a Hamiltonian deformation of the Riemann hierarchy with $K=\d_x$ is tau-symmetric if and only if it is special. 
\end{enumerate}
\end{remark}

\medskip

Normal Miura transformations preserve the property of being tau-symmetric. Indeed, for an arbitrary tau-symmetric deformation of the Riemann hierarchy and a normal Miura transformation~\eqref{eq:normal Miura}, the differential polynomials:
$$
\th_d=\left.\lb h_d+\d_x\{P,\oh_{d+1}\}_K\rb\right|_{u_l\mapsto\d_x^l u(\tu_*,\eps)},\quad d\ge -1,
$$
define a tau-structure for the transformed hierarchy. A converse statement is also true (see the second part of the following proposition).

\medskip

\begin{proposition}\label{proposition:tau-symmetric to special}{\ }
\begin{enumerate}
\item Any tau-symmetric deformation of the Riemann hierarchy can be transformed to a special Hamiltonian deformation by a suitable normal Miura transformation.

\smallskip

\item Consider an arbitrary tau-symmetric deformation of the Riemann hierarchy and a Miura transformation $\Phi\in\Mi^{(2)}$. Then the transformed hierarchy is tau-symmetric if and only if $\Phi$ is normal.
\end{enumerate}
\end{proposition}
\begin{proof}
1. Since $\{u,\int u dx\}_K=0$, we have $\Coef_{\d_x^0}K=0$, and therefore, by Proposition~\ref{proposition:K to dx by normal}, there exists a normal Miura transformation transforming the operator $K$ to $\d_x$. By Remark~\ref{remark:about deformations}, the resulting Hamiltonian deformation is special.

\medskip

2. We only have to prove the ``only if'' part. So we consider an arbitrary tau-symmetric deformation of the Riemann hierarchy and a Miura transformation $\Phi\in\Mi^{(2)}$ such that the transformed hierarchy is tau-symmetric. By Part~1, both the initial hierarchy and the transformed hierarchy can be transformed to special Hamiltonian deformations of the Riemann hierarchy by some normal Miura transformations, denote these Miura transformations by $\Phi_1,\Phi_2\in\Mi_\norm$, respectively. It is enough to prove that $\Phi_2\circ\Phi\circ\Phi_1^{-1}\in\Mi_\norm$. The Miura transformation $\Phi_2\circ\Phi\circ\Phi_1^{-1}$ transforms a special Hamiltonian transformation of the Riemann hierarchy to another special Hamiltonian deformation. So without loss of generality we may assume that our initial hierarchy is a special Hamiltonian deformation of the Riemann hierarchy and $\Phi\in\Mi^{(2)}_{\d_x}$.

\medskip

We know that our Miura transformation has the form $\cdots\circ\Phi_{\og_{k+2},\d_x}\circ\Phi_{\og_{k+1},\d_x}\circ\Phi_{\og_k,\d_x}$, $\og_i\in\Lambda_{u;i-1}$, $k\ge 2$. Then
$$
\oh_1[\tu]=\left.\lb\oh_1-\eps^k\{\oh_1,\og_k\}_{\d_x}\rb\right|_{u_l\mapsto\tu_l}+O(\eps^{k+1}).
$$
Since $\frac{\d h_1}{\d u}=\int\frac{u^2}{2}dx$ and $\frac{\d h_1[\tu]}{\d \tu}=\int\frac{\tu^2}{2}dx$, we have $\frac{\d}{\d u}\{\oh_1,\og_k\}_{\d_x}=0$. Therefore, 
$$
0=\frac{\d}{\d u}\{\oh_1,\og_k\}_{\d_x}=\left\{\int\frac{u^2}{2},\og_k\right\}_{\d_x}+\left\{\oh_1,\frac{\d}{\d u}\og_k\right\}_{\d_x}=\left\{\oh_1,\frac{\d}{\d u}\og_k\right\}_{\d_x}.
$$
and since $\deg\og_k\ge 1$ we obtain $\frac{\d}{\d u}\og_k=0$. Thus, $\Phi_{\og_k,\d_x}\in\Mi^{(2)}_{\d_x}\cap\Mi_\norm$. Applying the same argument to the Miura transformation $\cdots\circ\Phi_{\og_{k+2},\d_x}\circ\Phi_{\og_{k+1},\d_x}$ and so on, we conclude that $\frac{\d}{\d u}\og_i=0$ for all $i$ and hence $\Phi\in\Mi_\norm$.
\end{proof}

\medskip

Part 1 of the proposition immediately implies the following statement.

\begin{corollary}
The set of equivalence classes of tau-symmetric deformations of the Riemann hierarchy under normal Miura transformations is in one-to-one correspondence with the set of equivalence classes of special Hamiltonian deformations of the Riemann hierarchy under Miura transformations from the group $\Mi_{\d_x}\cap\Mi_\norm$. 
\end{corollary}

\medskip

\subsection{Dubrovin--Liu--Yang--Zhang conjecture and partial CohFTs}

\begin{definition}{\ }
\begin{enumerate}
\item We will say that a tau-symmetric deformation of the Riemann hierarchy is in the \emph{generalized standard form} if $K=\d_x$ and the Hamiltonian $\oh_1$ has the form
\begin{gather*}
\oh_1=\int\left(\frac{u^3}{6}+\eps^2 a u_x^2+\sum_{k\ge 4}\eps^k\sum_{\lambda\in\cP_k'}a_\lambda u_\lambda\right)dx,\quad a,a_\lambda\in\mbC.
\end{gather*}

\smallskip

\item We will say that a tau-symmetric deformation of the Riemann hierarchy is in the \emph{standard form} if it is in the generalized standard form where all the coefficients~$a_\lambda$ with odd~$|\lambda|$ vanish. This notion was introduced in~\cite{DLYZ16}.
\end{enumerate}
\end{definition}

\medskip

\begin{conjecture}[Conjecture 6.1 in~\cite{DLYZ16}]\label{conjecture:DLYZ}
Consider a tau-symmetric deformation of the Riemann hierarchy.
\begin{enumerate}
\item There exists a unique normal Miura transformation $u\mapsto\tu(u_*,\eps)$ such that the transformed hierarchy is in the standard form: $K_{\tu}=\d_x$,
$$
\oh_1[\tu]=\int\left(\frac{\tu^3}{6}+\eps^2 a \tu_x^2+\sum_{g\ge 2}\eps^{2g}\sum_{\lambda\in\cP_{2g}'}a_\lambda \tu_\lambda\right)dx,\quad a,a_\lambda\in\mbC.
$$

\smallskip

\item If $a=0$, then $a_\lambda=0$ for all $\lambda$.

\smallskip

\item If $a\ne 0$, then all the coefficients $a_\lambda$ are uniquely determined in terms of $a$ and $a_{(2^g)}$, $g\ge 2$.

\end{enumerate}
\end{conjecture}

\medskip

We can prove a weaker version of the first part of the conjecture.

\medskip

\begin{theorem}\label{theorem:generalized standard form}
For any tau-symmetric deformation of the Riemann hierarchy, there exists a unique normal Miura transformation $u\mapsto\tu(u_*,\eps)$ such that the transformed hierarchy is in the generalized standard form.
\end{theorem}
\begin{proof}
The uniqueness was already proved in~\cite[pages 431--432]{DLYZ16}.

\medskip

Let us prove the existence. By Proposition~\ref{proposition:tau-symmetric to special}, we can assume that our deformation is a special Hamiltonian deformation of the Riemann hierarchy,
$$
\oh_1=\int\left(\frac{u^3}{6}+\eps^2 a u_x^2+\sum_{k\ge 3}\eps^k\sum_{\lambda\in\cP_k^\circ}a_\lambda u_\lambda\right)dx,\quad a,a_\lambda\in\mbC,
$$
where the fact that the coefficients $a$ and $a_\lambda$ are constants follows from Corollary~\ref{corollary:simple}. Let $k_0$ be the minimal $k$ such that there exists a partition $\lambda\in\cP_k^\circ\backslash\cP_k'$ such that $a_\lambda\ne 0$. Then let 
\begin{align*}
&\tlambda:=\max\left\{\left.\lambda\in\cP_{k_0}^\circ\backslash\cP_{k_0}'\right|a_\lambda\ne 0\right\},\quad\tlambda=(\tlambda_1,\ldots,\tlambda_l,1),\\
&\oh:=-\frac{a_{\tlambda}}{k_0+l-m_1(\tlambda)-1}\int\prod_{i=1}^l u_{\tlambda_i}dx\in\Lambda_{u;k_0-1}.
\end{align*}
Using Lemma~\ref{lemma:differentiating ulambda}, we see that the Miura transformation $\Phi_{\oh,\d_x}\in\Mi_{\d_x}\cap\Mi_\norm$ doesn't change the coefficient of~$\eps^i$ with $i<k_0$ in~$\oh_1$, doesn't change the coefficients $a_\lambda$ with $|\lambda|=k_0$ and $\lambda>\tlambda$, and kills the coefficient $a_{\tlambda}$. Then we apply the same procedure, and after a finite number of steps we will kill all the coefficients $a_\lambda$ with $\lambda\in\cP_{k_0}^\circ\backslash\cP'_{k_0}$. Then we go to the coefficient of~$\eps^{k_0+1}$ and so on. The required Miura transformation from the group $\Mi_{\d_x}\cap\Mi_\norm$ is then the composition of the constructed Miura transformations.
\end{proof}

\medskip

The following theorem gives an explicit construction of tau-symmetric deformations of the Riemann hierarchy in the standard form with arbitrary coefficients $a\ne 0$ and $a_{(2^g)}$, $g\ge 2$, as the DR hierarchy of a family of partial CohFTs. Therefore, if Conjecture~\ref{conjecture:DLYZ} is true, then the theorem implies that the DR hierarchies of partial CohFTs give all standard forms of tau-symmetric deformations of the Riemann hierarchy.

\medskip

\begin{theorem}\label{theorem:DR for rank one partial CohFTs}
Let $G\in\mbC^*$ and $s_1,s_2,s_3,\ldots\in\mbC$. Consider the partial CohFT
\begin{gather}\label{eq:rank 1 partial CohFTs}
c_{g,n}=G^g \exp\lb\sum_{i\ge 1}s_i\Ch_{2i-1}(\mbE)\rb\in H^*(\oM_{g,n}).
\end{gather}
\begin{enumerate}
\item The associated DR hierarchy is a special Hamiltonian deformation of the Riemann hierarchy in the standard form,
\begin{gather}\label{eq:h1 for DR for partial CohFT}
\og_1=\int\left(\frac{u^3}{6}-\frac{\eps^2 G}{24} u_x^2+\sum_{g\ge 2}\eps^{2g}G^g\sum_{\lambda\in\cP_{2g}'}a_\lambda u_\lambda\right)dx,\quad a_\lambda\in\mbC.
\end{gather}

\smallskip

\item For $\lambda\in\cP_{2g}'$, $g\ge 2$, the coefficient $a_\lambda$ is a polynomial in $s_1,\ldots,s_{\lfloor\frac{g-2+l(\lambda)}{2}\rfloor}$ with rational coefficients of degree $g-3+l(\lambda)$, where $\deg s_i:=2i-1$. Moreover, $a_{(2^g)}$ has the form
$$
a_{(2^g)}=(-1)^g\frac{(3g-2)|B_{2g-2}||B_{2g}|}{4g((2g-2)!)^2}s_{g-1}+T_g(s_1,\ldots,s_{g-2}),\quad T_g\in\mbQ[s_1,\ldots,s_{g-2}].
$$
\end{enumerate}
\end{theorem}
\begin{proof}
{\it 1}. For $G=1$, Part 1 was proved in~\cite[Proposition~8.3]{BDGR20}. It is clear that the multiplication of $c_{g,n}$ by $G^g$ corresponds to the rescaling $\eps\mapsto\sqrt{G}\eps$.

\medskip   

{\it 2}. Without loss of generality, we can assume that $G=1$. For $\lambda\in\cP_{2g}'$, $g\ge 2$, $l=l(\lambda)$, we have (see the proof of Proposition~8.3 in~\cite{BDGR20})
\begin{gather}\label{eq:formula for alambda}
a_\lambda=\frac{2g-2+l}{(2g-2)|\Aut(\lambda)|}\Coef_{a_1^{\lambda_1}\cdots a_l^{\lambda_l}}\left[\int_{\oM_{g,1}}\psi_1\lambda_g\pi_*\lb\DR_g(0,a_1,\ldots,a_l)\rb \exp\lb\sum_{i\ge 1}s_i\Ch_{2i-1}(\mbE)\rb\right],
\end{gather}
where $\pi\colon\oM_{g,l+1}\to\oM_{g,1}$ is the forgetful map that forgets the last $l$ marked points. A simple degree counting shows that the class $\exp\lb\sum_{i\ge 1}s_i\Ch_{2i-1}(\mbE)\rb$ in formula~\eqref{eq:formula for alambda} can be replaced by the class $\left[\exp\lb\sum_{i\ge 1}s_i\Ch_{2i-1}(\mbE)\rb\right]_{g-3+l}$, which is clearly a cohomology class on $\oM_{g,1}$ depending polynomially on $s_1,\ldots,s_{\lfloor\frac{g-2+l(\lambda)}{2}\rfloor}$ and with the required homogeneity property.

\medskip 

Regarding the coefficient $a_{(2^g)}$, by~\cite[Proposition~8.3]{BDGR20}, we have
$$
a_{(2^g)}=(3g-2)\int_{\oM_{g}}\lambda_g\left[\exp\lb\sum_{i\ge 1}s_i\Ch_{2i-1}(\mbE)\rb\right]_{2g-3}.
$$
Clearly the class $\left[\exp\lb\sum_{i\ge 1}s_i\Ch_{2i-1}(\mbE)\rb\right]_{2g-3}$ is the sum of the class $s_{g-1}\Ch_{2g-3}(\mbE)$ and a cohomology class depending polynomially on $s_1,\ldots,s_{g-2}$. An elementary computation with symmetric functions shows that
$$
\lambda_g\Ch_{2g-3}(\mbE)=\frac{(-1)^{g-2}}{(2g-3)!}\lambda_g\lambda_{g-1}\lambda_{g-2},
$$
and then, using the formula (see e.g.~\cite[equation~(9)]{FP00})
$$
\int_{\oM_g}\lambda_g\lambda_{g-1}\lambda_{g-2}=\frac{1}{2(2g-2)!}\frac{|B_{2g-2}|}{2g-2}\frac{|B_{2g}|}{2g},
$$
we complete the proof.
\end{proof}

\medskip


\section{Hierarchies of conservation laws and F-CohFTs}\label{section:conservation laws}

\subsection{Arsie--Lorenzoni--Moro conjecture}

Consider a PDE 
\begin{gather}\label{eq:first flow}
\frac{\d u}{\d t}=Q,\quad Q=uu_x+O(\eps)\in\hcA_{u;1}.
\end{gather}
Let us expand
$$
Q=\sum_{k\ge 0}\eps^k\sum_{\lambda\in\cP_{k+1}}b_\lambda(u)u_\lambda,\quad b_\lambda(u)\in\mbC[[u]].
$$
In~\cite{ALM15b}, the authors noticed that the collection of formal power series $b_{k}(u)$, $k\ge 2$, is invariant under Miura transformations. Suppose now that a PDE~\eqref{eq:first flow} is a part of a deformation of the Riemann hierarchy $\frac{\d u}{\d t_d}=Q_d$, $d\ge 0$, $Q_1=Q$. We know that the whole deformation is uniquely determined by $Q_1$. In~\cite[Conjecture~1.1]{ALM15b}, the authors conjectured that the equivalence class of this deformation under Miura transformations is uniquely determined by the formal power series $b_{k}(u)$, $k\ge 2$.  

\medskip

Consider now a PDE of the form 
\begin{gather}\label{eq:first flow,conservation law}
\frac{\d u}{\d t}=\d_x P,\quad P\in\hcA_{u;0}.
\end{gather}
Note that if another PDE $\frac{\d u}{\d s}=Q$ with $Q\in\hcA_{u;1}$ commutes with it, then by Lemma~\ref{lemma:uniqueness} we have $Q=\d_x\tQ$ for some $\tQ\in\hcA_{u;0}$. Clearly the form~\eqref{eq:first flow,conservation law} is preserved under the Miura transformations 
\begin{gather}\label{eq:Miura for conservation laws}
u\mapsto v(u_*,\eps)=u+\d_x f,\quad f\in\hcA_{u;-1}.
\end{gather}
Suppose now that $P=\frac{u^2}{2}+O(\eps)$. By~\cite[Theorem~3.2]{ALM15a}, there exists a unique Miura transformation~\eqref{eq:Miura for conservation laws} such that the transformed equation $\frac{\d v}{\d t_d}=\d_x\tP$ has the form
$$
\tP=\frac{v^2}{2}+\eps a(v)v_x+\sum_{k\ge 2}\eps^k\omega_k(v_*),\quad\text{where}\quad \frac{\d\omega_k}{\d v_x}=0.
$$
The resulting equation is called the \emph{normal form} of the initial equation.

\medskip

\begin{conjecture}[Section 3 in \cite{ALM15b}]\label{ALM-conjecture}
Consider a deformation of the Riemann hierarchy 
\begin{gather}\label{eq:ALM-normal form-1}
\frac{\d u}{\d t_d}=\d_x P_d,\quad P_d=\frac{u^{d+1}}{(d+1)!}+O(\eps)\in\hcA_{u;0},\quad d\ge 0,
\end{gather}
where the flow $\frac{\d}{\d t_1}$ is in the normal form:
\begin{gather}\label{eq:ALM-normal forma-2}
P_1=\frac{u^2}{2}+\eps a(u)v_x+\sum_{k\ge 2}\eps^k\sum_{\lambda\in\cP_k,\,\lambda_i\ge 2}c_\lambda(u)u_\lambda,\quad c_\lambda(u)\in\mbC[[u]].
\end{gather}
Suppose that $a(u)=0$ and $c_2(u)\ne 0$. Then the following is true.
\begin{enumerate}
\item $c_\lambda(u)=0$ for all $\lambda$ with odd $|\lambda|$, and therefore all the deformation is even.

\smallskip

\item The deformation is uniquely determined by the formal power series $c_2(u),c_4(u),\ldots$.
\end{enumerate}
\end{conjecture}

\medskip

\begin{remark}
Actually, the authors of~\cite{ALM15b} conjectured a stronger result. Consider a deformation of the Riemann hierarchy $\frac{\d u}{\d t_d}=\d_x P_d$ with
$$
P_1=\frac{u^2}{2}+\sum_{k\ge 2}\eps^k\sum_{\lambda\in\cP_k,\,\lambda_i\ge 2}c_\lambda(u)u_\lambda.
$$ 
We know that the whole deformation is uniquely determined by $P_1$. The authors of~\cite{ALM15b} showed that the existence of a deformation with a given $P_1$ is equivalent to an infinite system of polynomial equations for the coefficients $c_\lambda(u)$ and their $u$-derivatives $c_{\lambda}^{(n)}(u)$. The authors of~\cite{ALM15b} conjectured that that if $c_2(u)\ne 0$, then this system can be solved with $c_2(u),c_4(u),\ldots$ being independent functional parameters, where any other coefficient $c_{\lambda}(u)$ can be expressed as a polynomial in $c^{(l)}_{2g}(u)$, $g\ge 1$, $l\ge 0,$ and~$c_2(u)^{-1}$. In particular, analyzing a part of the infinite system, the authors of~\cite{ALM15b} found the following formula for $c_{2,2}$:
$$
c_{2,2}=\frac{1}{144}\frac{1}{c_2^2}\left[334c_2^3c_2^{(2)}-168c_2^2(c_2^{(1)})^2+335 c_4 c_2^{(1)}-330c_2^2c_4^{(1)}+280c_2c_6-400c_4^2\right].
$$ 
\end{remark}

\medskip

\begin{remark}
Comparing Conjectures~\ref{conjecture:DLYZ} and~\ref{ALM-conjecture}, one can note that in Conjecture~\ref{ALM-conjecture} there is no analog of Part~2 of Conjecture~\ref{conjecture:DLYZ}. One could expect that if $a(u)=c_2(u)=0$, then $c_\lambda(u)=0$ for all $\lambda$. However, the paper~\cite{ALM15b} doesn't contain such a claim. Additionally, on the conference ``New Trends in Moduli, Integrability and Deformations'' (9--12 June 2025, Padova, Italy) P. Lorenzoni informed us that they don't have enough evidence to claim this.  
\end{remark}

\medskip

\subsection{The DR hierarchies of F-CohFTs in the framework of the Arsie--Lorenzoni--Moro conjecture}

The following theorem gives an explicit construction of deformations of the Riemann hierarchy~\eqref{eq:ALM-normal form-1} where the flow $\frac{\d}{\d t_1}$ is in the normal form~\eqref{eq:ALM-normal forma-2} with $a(u)=0$, vanishing $c_\lambda(u)$ for odd $|\lambda|$, and arbitrary constant coefficients $c_2\ne 0,c_4,c_6,\ldots$, as the DR hierarchies of a family of F-CohFTs. Therefore, if Conjecture~\ref{ALM-conjecture} is true, then the theorem implies that all normal forms of deformations of the Riemann hierarchy~\eqref{eq:ALM-normal form-1} with $a(u)=0$ and constant coefficients $c_2\ne 0,c_4,c_6,\ldots$ are obtained as DR hierarchies of F-CohFTs.

\medskip

\begin{theorem}\label{theorem:DR hierarchy for F-CohFT}
Let $G\in\mbC^*$ and $R(z)=\exp\lb\sum_{i\ge 1}r_i z^i\rb$, $r_i\in\mbC$. Consider the DR hierarchy corresponding to the F-CohFT $\{R.c^{\triv,G}_{g,n+1}\}$:
$$
\frac{\d u}{\d t_d}=\d_x P_d,\quad d\ge 0.
$$
\begin{enumerate}
\item The flow $\frac{\d}{\d t_1}$ is in the normal form and, moreover,
$$
P_1=\frac{u^2}{2}+\sum_{g\ge 1}\eps^{2g}G^g\sum_{\lambda\in\cP_{2g},\,\lambda_i\ge 2}c_\lambda u_\lambda,
$$
where $c_\lambda$ is a polynomial in $r_1,\ldots,r_{g+l(\lambda)-2}$ with rational coefficients of degree $g+l(\lambda)-2$, where $\deg r_i:=i$.

\smallskip

\item The coefficients $c_{2g}$, $c_{(2^g)}$, and $c_{(4,2^{g-2})}$ have the form
\begin{align*}
& c_{2g}=
\begin{cases}
\frac{1}{12},&\text{if $g=1$},\\
\alpha_g r_{g-1}+\tT_{2g}(r_1,\ldots,r_{g-2}),&\text{if $g\ge 2$},
\end{cases}\\
&c_{(2^g)}=\beta_g r_{2g-2}+\tT_{(2^g)}(r_1,\ldots,r_{2g-3}),\quad c_{(4,2^{g-2})}=\gamma_g r_{2g-3}+\tT_{(4,2^{g-2})}(r_1,\ldots,r_{2g-4}),\quad g\ge 2,
\end{align*}
where 
\begin{align*}
\alpha_g=\frac{2g}{4^g}\frac{1}{(2g+1)!!},\qquad \beta_g=(3g-1)(2g-1)\frac{|B_{2g}|}{(2g)!},\qquad \gamma_g=\frac{(3g-2)|B_{2g}|}{8(2g-3)!}, 
\end{align*}
and $\tT_{2g},\tT_{(2^g)}$, $\tT_{(4,2^{g-2})}$ are some polynomials with rational coefficients.
\end{enumerate}
\end{theorem}
\begin{proof}
It is clear that $R.c^{\triv,G}_{g,n+1}=G^g R.c^{\triv,1}_{g,n+1}$, and so on the level of the DR hierarchy going from the F-CohFT $\{R.c^{\triv,1}_{g,n+1}\}$ to the F-CohFT $\{R.c^{\triv,G}_{g,n+1}\}$ corresponds to the rescaling $\eps\mapsto\sqrt{G}\eps$. Therefore, without loss of generality we can assume that $G=1$.

\medskip

{\it 1}. The proof of the fact that the flow $\frac{\d}{\d t_1}$ is in the normal form is similar to the proof of~\cite[Proposition~8.3]{BDGR20}. We have
\begin{gather*}
P_1:=\sum_{g,n\ge 0}\frac{\eps^{2g}}{n!}\sum_{\substack{\od=(d_1,\ldots,d_n)\in\mbZ_{\ge 0}^n\\d_1+\ldots+d_n=2g}}\underbrace{\Coef_{a_1^{d_1}\cdots a_n^{d_n}}\lb\int_{\oM_{g,n+2}}\hspace{-0.3cm}\psi_2\lambda_g\DR_g\lb-\sum a_i,0,a_1,\ldots,a_n\rb R.c^{\triv,1}_{g,n+2}\rb}_{c_{\od}:=} u_{\od},
\end{gather*}
where $u_{\od}:=\prod_{i=1}^n u_{d_i}$, and 
\begin{align*}
&\int_{\oM_{g,n+2}}\psi_2\lambda_g\DR_g\lb-\sum a_i,0,a_1,\ldots,a_n\rb R.c^{\triv,1}_{g,n+2}=\\
&\hspace{2cm}=(2g-1+n)\int_{\oM_{g,n+1}}\lambda_g\DR_g\lb-\sum a_i,a_1,\ldots,a_n\rb R.c^{\triv,1}_{g,n+1}=\\
&\hspace{2cm}=\begin{cases}
\delta_{n,2},&\text{if $g=0$},\\
(2g-1+n)\int_{\oM_{g,1}}\lambda_g\pi_{n*}\lb\DR_g\lb-\sum a_i,a_1,\ldots,a_n\rb\rb R.c^{\triv,1}_{g,1},&\text{if $g\ge 1$},
\end{cases}
\end{align*}
where $\pi_k\colon\oM_{g,k+l}\to\oM_{g,l}$ is the map that forgets the last $k$ marked points. Consider the case $g\ge 1$. Since the polynomial class
$$
\left.\pi_{1*}\lb\DR_g\lb-\sum a_i,a_1,\ldots,a_m\rb\rb\right|_{\cM_{g,m}^\ct}\in H^{2g-2}(\cM_{g,m}^\ct,\mbQ),\quad m\ge 1,
$$
is divisible by $a_m^2$ \cite[Lemma~5.1]{BDGR18}, we conclude that $c_{\od}=0$ unless $d_1,\ldots,d_n\ge 2$. Moreover, a simple degree counting gives that 
\begin{multline*}
\int_{\oM_{g,1}}\lambda_g\pi_{n*}\lb\DR_g\lb-\sum a_i,a_1,\ldots,a_n\rb\rb R.c^{\triv,1}_{g,1}=\\
=\int_{\oM_{g,1}}\lambda_g\pi_{n*}\lb\DR_g\lb-\sum a_i,a_1,\ldots,a_n\rb\rb\left[R.c^{\triv,1}_{g,1}\right]_{g+n-2},
\end{multline*}
and it remains to note that $\left[R.c^{\triv,1}_{g,n}\right]_k$ is a cohomology class on $\oM_{g,n}$ depending polynomially on $r_1,\ldots,r_k$ and with the required homogeneity property.

\medskip

{\it 2}. Using the formulas for the $R$-matrix action on F-CohFTs from~\cite[Section~4]{ABLR23}, we compute, for any $l\ge 1$,
\begin{gather*}
\left.\left[R.c^{\triv,1}_{g,n+1}\right]_l\right|_{r_1=r_2=\ldots=r_{l-1}=0}=r_l\Bigg(\kappa_l+(-1)^l\psi_1^l-\sum_{i=2}^{n+1}\psi_i^l+\sum_{\substack{i+j=l-1\\g_1+g_2=g\\I\sqcup J=[n+1]\backslash\{1\}}}(-1)^j\tikz[baseline=-1mm]{\coordinate (A) at (0,0);\coordinate (B) at (12mm,0);\draw (A)--(B);\legm{A}{180}{1};\leg{A}{-60};\leg{A}{-120};\leg{B}{-60};\leg{B}{-120};\gg{g_1}{A};\gg{g_2}{B};\lab{A}{28}{4.4mm}{\psi^i};\lab{B}{146}{4mm}{\psi^j};\lab{A}{-90}{5.9mm}{...};\lab{A}{-90}{7mm}{\underbrace{\phantom{aaa}}_{\text{$I$}}};\lab{B}{-90}{5.9mm}{...};\lab{B}{-90}{7mm}{\underbrace{\phantom{aaa}}_{\text{$J$}}};}\Bigg),
\end{gather*}
where we used the standard graphical notation for tautological cohomology classes on $\oM_{g,n}$, see e.g.~\cite[Section~2.1]{BGR19}. Let us denote the class in the brackets by $\theta_{g,n+1,l}$. 

\medskip

Regarding the coefficient $c_{2g}$, we have
\begin{align*}
c_{2g}=&\Coef_{a^{2g}}\int_{\oM_{g,3}}\lambda_g\DR_g(a,0,-a)\psi_2\left[R.c^{\triv,1}_{g,3}\right]_{g-1}=\\
=&2g\,\Coef_{a^{2g}}\int_{\oM_{g,2}}\lambda_g\DR_g(a,-a)\left[R.c^{\triv,1}_{g,2}\right]_{g-1}.
\end{align*}
Since 
\begin{gather}\label{eq:lambda1-DR1}
\int_{\oM_{1,2}}\lambda_1\DR_1(a,-a)=\frac{a^2}{24},
\end{gather}
we obtain $c_2=\frac{1}{12}$.

\medskip

\underline{Let us compute $\alpha_g$}, $g\ge 2$. We have
$$
\alpha_g=2g\,\Coef_{a^{2g}}\int_{\oM_{g,2}}\lambda_g\DR_g(a,-a)\theta_{g,2,g-1}.
$$

\medskip

\begin{lemma}
We have
$$
\sum_{g\ge 0}\int_{\oM_{g,3}}\DR_g(-a-b,a,b)\psi_1^g\lambda_g=e^{\frac{(a+b)^2}{24}}\sum_{n\ge 0}\frac{(-1)^n}{(2n+1)!!}\left(\frac{ab}{4}\right)^n.
$$
\end{lemma}
\begin{proof}
The proof is based on \cite[Theorem~4]{BSSZ15}. First, using this theorem and formula~\eqref{eq:lambda1-DR1}, by induction we obtain
$$
\int_{\oM_{g,2}}\psi_1^{g-1}\lambda_g\DR_g(a,-a)=\frac{a^{2g}}{24^g g!},\quad g\ge 1.
$$
Then denote
$$
P_g(a,b):=\int_{\oM_{g,3}}\DR_g(-a-b,a,b)\psi_1^g\lambda_g,\quad g\ge 0.
$$
Using \cite[Theorem~4]{BSSZ15}, we obtain the following recursion:
$$
P_g(a,b)=\frac{1}{2g+1}\frac{(a+b)^{2g}}{24^g g!}+\frac{a^2-ab+b^2}{12(2g+1)}P_{g-1}(a,b),\quad g\ge 1,
$$
with $P_0(a,b)=1$. This recursion can be rewritten as the following equality for the generating series $F:=\sum_{g\ge 0}P_g(a,b)z^{2g}$:
$$
\frac{\d}{\d z}(z F)=e^{\frac{z^2(a+b)^2}{24}}+z^2\frac{a^2-ab+b^2}{12}F.
$$
Introducing $\tF:=e^{-\frac{z^2(a+b)^2}{24}}F$, we obtain the following equality:
$$
\frac{\d}{\d z}(z\tF)=1-z^2\frac{ab}{4}\tF,
$$
which immediately implies that 
$$
\tF=\sum_{n\ge 0}\frac{(-1)^n}{(2n+1)!!}z^{2n}\left(\frac{ab}{4}\right)^n.
$$
\end{proof}

\medskip

Using this lemma, we compute
\begin{align*}
&2g\,\Coef_{a^{2g}}\int_{\oM_{g,2}}\lambda_g\DR_g(a,-a)\kappa_{g-1}=2g\,\Coef_{a^{2g}}\int_{\oM_{g,1}}\lambda_g\psi_1^g\DR_g(0,a,-a)=\frac{2g}{4^g}\frac{1}{(2g+1)!!},\\
&\Coef_{a^{2g}}\int_{\oM_{g,2}}\lambda_g\DR_g(a,-a)\Bigg((-1)^{g-1}\psi_1^{g-1}-\psi_2^{g-1}+\sum_{\substack{i+j=g-2\\g_1+g_2=g}}(-1)^j\tikz[baseline=-1mm]{\coordinate (A) at (0,0);\coordinate (B) at (12mm,0);\draw (A)--(B);\legm{A}{180}{1};\legm{B}{0}{2};\gg{g_1}{A};\gg{g_2}{B};\lab{A}{28}{4.4mm}{\psi^i};\lab{B}{146}{4mm}{\psi^j};}\Bigg)=\\
&\hspace{7.75cm}=\frac{1}{24^g}\lb\frac{(-1)^{g-1}}{g!}-\frac{1}{g!}+\sum_{\substack{i+j=g\\i,j\ge 1}}\frac{(-1)^{j+1}}{i!j!}\rb=0.
\end{align*}
Note that the class $\theta_{g,2,g-1}$ also contains the sum
$$
\sum_{\substack{i+j=g-2\\g_1+g_2=g}}(-1)^j\tikz[baseline=-1mm]{\coordinate (A) at (0,0);\coordinate (B) at (12mm,0);\draw (A)--(B);\legm{A}{180}{1};\legm{A}{-90}{2};\gg{g_1}{A};\gg{g_2}{B};\lab{A}{28}{4.4mm}{\psi^i};\lab{B}{146}{4mm}{\psi^j};},
$$
which however vanishes when we multiply it by $\lambda_g\DR_g(a,-a)$. The required formula for $\alpha_g$ is proved.

\medskip

\underline{Let us compute $\beta_g$}, $g\ge 2$. We have
\begin{align*}
c_{(2^g)}=&\frac{1}{g!}\Coef_{a_1^2\cdots a_g^2}\int_{\oM_{g,g+2}}\DR_g\lb-\sum a_i,0,a_1,\ldots,a_g\rb\lambda_g\psi_2\left[R.c^{\triv,1}_{g,g+2}\right]_{2g-2}=\\
=&\frac{3g-1}{g!}\Coef_{a_1^2\cdots a_g^2}\int_{\oM_{g,g+1}}\DR_g\lb-\sum a_i,a_1,\ldots,a_g\rb\lambda_g\left[R.c^{\triv,1}_{g,g+1}\right]_{2g-2}=\\
=&(3g-1)\int_{\oM_{g,1}}\lambda_g\left[R.c^{\triv,1}_{g,1}\right]_{2g-2},
\end{align*}
and therefore
\begin{align*}
\beta_g=&(3g-1)\int_{\oM_{g,1}}\lambda_g\theta_{g,1,2g-2}=\\
=&(3g-1)\int_{\oM_{g,1}}\lambda_g\Bigg(\kappa_{2g-2}+\psi_1^{2g-2}+\sum_{\substack{i+j=2g-3\\g_1+g_2=g}}(-1)^j\tikz[baseline=-1mm]{\coordinate (A) at (0,0);\coordinate (B) at (12mm,0);\draw (A)--(B);\legm{A}{180}{1};\gg{g_1}{A};\gg{g_2}{B};\lab{A}{28}{4.4mm}{\psi^i};\lab{B}{146}{4mm}{\psi^j};}\Bigg)=\\
=&(3g-1)\sum_{\substack{g_1+g_2=g\\g_1,g_2\ge 0}}b_{g_1}b_{g_2},
\end{align*}
where $b_h:=\begin{cases}1,&\text{if $h=0$},\\\int_{\oM_{h,1}}\lambda_h\psi_1^{2h-2},&\text{if $h\ge 1$}.\end{cases}$. We have $b_h=\frac{2^{2h-1}-1}{2^{2h-1}}\frac{|B_{2h}|}{(2h)!}$ for $h\ge 1$ \cite[Theorem~1]{FP03}. Using the identity $\sum_{\substack{g_1+g_2=g\\g_1,g_2\ge 0}}b_{g_1}b_{g_2}=\frac{(2g-1)|B_{2g}|}{(2g)!}$ \cite[Remark~5.1]{Bur15}, we complete the proof of the required formula for~$\beta_g$.

\medskip

\underline{Let us compute $\gamma_g$}, $g\ge 2$. Since we are interested in the coefficient of $r_{2g-3}$ in the polynomial~$c_{(4,2^{g-2})}$, we can set $r_{2i}=0$ for all $i\ge 1$. Then (see, e.g., \cite[Section~1.1]{FP00-Hodge-and-Gromov})
$$
\lambda_g R.c^{\triv,1}_{g,n+1}=\lambda_g\exp\lb\sum_{i\ge 1}s_i\Ch_{2i-1}(\mbE)\rb,
$$
where $s_i=r_{2i-1}\frac{(2i)!}{B_{2i}}$. So our hierarchy is exactly the hierarchy considered in Theorem~\ref{theorem:DR for rank one partial CohFTs}. We have $c_{(4,2^{g-2})}=g(g-1)a_{(2^g)}$, and using Part~2 of this theorem we obtain
$$
\gamma_g=g(g-1)\frac{(2g-2)!}{B_{2g-2}}(-1)^g\frac{(3g-2)|B_{2g-2}||B_{2g}|}{4g((2g-2)!)^2}=\frac{(3g-2)|B_{2g}|}{8(2g-3)!},
$$ 
as required.
\end{proof}

\medskip

\begin{remark}\label{remark about ALM-conjecture}{\ }
\begin{enumerate}
\item We know that the DR hierarchy of the F-CohFT $\{R.c^{\triv,G}_{g,n+1}\}$ possesses a family of conserved quantities of the form 
$$
\og_d=\int\left(\frac{u^{d+2}}{(d+2)!}+O(\eps)\right)\in\hLambda_{u;0},\quad d\ge 0.
$$
We come to the following remarkable observation: assuming that Conjecture~\ref{ALM-conjecture} is true, if all the coefficients $c_{2g}(u)$, $g\ge 1$ (in the deformation of the Riemann hierarchy considered in the conjecture) are constants, then the deformation possesses an infinite family of conserved quantities. 

\smallskip

\item Consider the class of hierarchies given by the DR hierarchies of the F-CohFTs $\{R.c^{\triv,G}_{g,n+1}\}$. The theorem implies that this class can be parameterized in two different ways. The first parameterization is given by the numbers $c_{2g}$, $g\ge 1$. The second parameterization is given by the numbers $c_2$, $c_{(2^g)}$, $c_{(4,2^{g-2})}$, $g\ge 2$. The second parameterization has the following property: for any fixed $g_0\ge 1$, the parameters found at the approximation up to~$\eps^{2g_0}$ fully determine the deformation at the same approximation. The first parameterization doesn't have this property: in order to find the deformation at the approximation up to $\eps^{2g_0}$, one has to know the parameters $c_2,c_4,\ldots,c_{2(2g_0-1)}$, which are found at the approximation up to $\eps^{2(2g_0-1)}$.
\end{enumerate}
\end{remark}

\medskip

Consider again the DR hierarchy from Theorem~\ref{theorem:DR hierarchy for F-CohFT}. We know that if $r_{2i}=0$ for all $i\ge 1$, then the F-CohFT $\{R.c^{\triv,G}_{g,n+1}\}$ is associated to a partial CohFT, and therefore the DR hierarchy is actually a Hamiltonian deformation of the Riemann hierarchy. The next proposition describes how to reformulate the vanishing of the paramaters $r_{2i}$ in terms of the equation for the flow~$\frac{\d}{\d t_1}$. 

\begin{proposition}
Consider the DR hierarchy from Theorem~\ref{theorem:DR hierarchy for F-CohFT}. Then we have $c_{(2^h)}=0$ for all $h\ge 2$ if and only if $r_{2g}=0$ for all $g\ge 1$.
\end{proposition}
\begin{proof}
Let us prove the ``if'' part. Suppose that $r_{2g}=0$ for all $g\ge 1$. Then our F-CohFT is associated to a partial CohFT, and therefore our hierarchy is Hamiltonian,
$$
P_d=\frac{\delta\og_d}{\delta u},\quad d\ge 0,
$$ 
where $\og_d$ is the Hamiltonian of the DR hierarchy. We know that $\og_1$ has the form~\eqref{eq:h1 for DR for partial CohFT}. It remains to note that for any $g\ge 2$ and $\lambda\in\cP_{2g}'$ we have
$$
\Coef_{u_{xx}^g}\lb\frac{\delta}{\delta u}u_\lambda\rb=0.
$$

\medskip

From the ``if'' part of the proposition and Part 2 of Theorem~\ref{theorem:DR hierarchy for F-CohFT}, it follows that the polynomial $\tT_{(2^g)}(r_1,\ldots,r_{2g-3})$ has the property
$$
\tT_{(2^g)}(r_1,\ldots,r_{2g-3})|_{r_{2i}=0}=0,\quad g\ge 2.
$$

\medskip

Let us now prove the ``only if'' part. Suppose that $r_{2h}\ne 0$ for some $h\ge 1$, and let $h_0$ be the minimal such $h$. By Part~2 of Theorem~\ref{theorem:DR hierarchy for F-CohFT}, we have
$$
c_{(2^{h_0+1})}=\underbrace{\beta_{h_0+1}r_{2h_0}}_{\ne 0}+\tT_{(2^{h_0+1})}(r_1,\ldots,r_{2h_0-1}).
$$
Since $r_2=r_4=\ldots=r_{2h_0-2}=0$, we obtain $\tT_{(2^{h_0+1})}(r_1,\ldots,r_{2h_0-1})=0$ and therefore $c_{(2^{h_0+1})}\ne 0$. This completes the proof of the proposition.
\end{proof}

\medskip

\subsection{A relation between the two conjectures}

\begin{theorem}
Conjecture~\ref{ALM-conjecture} implies Parts~1 and~3 of Conjecture~\ref{conjecture:DLYZ}.
\end{theorem}
\begin{proof}
Consider a tau-symmetric deformation of the Riemann hierarchy in the generalized standard form,
$$
\frac{\d u}{\d t_d}=\d_x\underbrace{\frac{\delta\oh_d}{\delta u}}_{P_d:=},\quad\oh_1=\int\left(\frac{u^3}{6}+\eps^2 a u_x^2+\sum_{k\ge 4}\eps^k\sum_{\lambda\in\cP_k'}b_\lambda u_\lambda\right)dx,\quad a\in\mbC^*,\quad b_\lambda\in\mbC.
$$
One can easily see that the flow $\frac{\d}{\d t_1}$ is in the normal form,
$$
P_1=\frac{u^2}{2}-2a\eps^2u_{xx}+\sum_{k\ge 4}\eps^k\sum_{\lambda\in\cP_k,\,\lambda_i\ge 2}c_\lambda u_\lambda,\quad c_\lambda\in\mbC,
$$
and that
\begin{gather}\label{eq:conjecture2 implies conjecture1-1}
c_{(2^g)}=0,\quad c_{(4,2^{g-2})}=g(g-1)b_{(2^g)},\quad g\ge 2.
\end{gather}
By Conjecture~\ref{ALM-conjecture}, all the coefficients $c_\lambda$ with odd $|\lambda|$ vanish, and therefore all the coefficients~$b_\lambda$ with odd $|\lambda|$ vanish. By Conjecture~\ref{ALM-conjecture} and Theorem~\ref{theorem:DR hierarchy for F-CohFT}, our hierarchy coincides with the DR-hierarchy corresponding to some F-CohFT $\{R.c^{\triv,-24a}_{g,n+1}\}$, $R(z)=\exp\lb\sum_{i\ge 1}r_i z^i\rb$, $r_i\in\mbC$, and therefore it is uniquely determined by the coefficients $a$, $c_{(2^g)}$, $c_{(4,2^{g-2})}$, $g\ge 2$. From~\eqref{eq:conjecture2 implies conjecture1-1} we conclude that the hierarchy is uniquely determined by the coefficients $a$, $b_{(2^g)}$, $g\ge 2$, which completes the proof of the theorem.
\end{proof}

\medskip

\subsection{On the Hamiltonian structure of the DR hierarchies of F-CohFTs}

\begin{theorem}
Let $G\in\mbC^*$ and $R(z)=\exp\lb\sum_{i\ge 1}r_i z^i\rb$, $r_i\in\mbC$. Then the DR hierarchy corresponding to the F-CohFT $\{R.c^{\triv,G}_{g,n+1}\}$ is a Hamiltonian deformation of the Riemann hierarchy if and only if $r_{2i}=0$ for all $i\ge 1$.
\end{theorem}
\begin{proof}
Regarding the ``if'' part, as it was already mentioned above, if $r_{2i}=0$ for all $i\ge 1$, then the F-CohFT $\{R.c^{\triv,G}_{g,n+1}\}$ comes from a partial CohFT, and therefore the corresponding DR hierarchy is a Hamiltonian deformation of the Riemann hierarchy.

\medskip

Let us prove the ``only if'' part. Suppose that the DR hierarchy is a Hamiltonian deformation of the Riemann hierarchy, with Poisson operator $K=\d_x+O(\eps)$, $\deg K=1$:
$$
\d_x P_d=K\frac{\delta\oh_d}{\delta u},\quad \oh_d=\int\lb\frac{u^{d+2}}{(d+2)!}+O(\eps)\rb dx\in\hLambda_{u;0},\quad d\ge 0.
$$
We have to prove that $r_{2i}=0$ for all $i\ge 1$. Without loss of generality, we can assume that $G=1$. We have $D_{\d_x P_1}(\oh_d)=D_{K\frac{\delta\oh_1}{\delta u}}(\oh_d)=\{\oh_d,\oh_1\}_K=0$. On the other hand, for the conserved quantity $\og_d$ of the DR hierarchy, we have $D_{\d_x P_1}(\og_d)=0$. Using Lemma~\ref{lemma:uniqueness}, we conclude that $\oh_d=\og_d$. So we have
\begin{gather}\label{eq:proving that DR is not Hamiltonian}
\d_x P_d=K\frac{\delta\og_d}{\delta u},\quad d\ge 0.
\end{gather}

\medskip

Let us apply the operator $\frac{\d}{\d u}$ to both sides of~\eqref{eq:proving that DR is not Hamiltonian}. Using that 
$$
\frac{\d P_d}{\d u}=P_{d-1},\quad \frac{\d\og_d}{\d u}=\og_{d-1},\quad d\ge 1,
$$
we obtain 
$$
\frac{\d K}{\d u}\frac{\delta\og_d}{\delta u}=0,\quad d\ge 1,
$$
which implies that $\frac{\d K}{\d u}=0$. 

\medskip

Define $R'(z):=\exp\lb\sum_{i\ge 1}r_{2i-1}z^{2i-1}\rb$, and consider the partial CohFT $\{(R'.c^{\triv,G=1})_{g,n}\}$ and the associated DR hierarchy
$$
\frac{\d u}{\d t_d}=\d_x\underbrace{\frac{\delta\og_d'}{\delta u}}_{P'_d:=},\quad d\ge 0.
$$
Suppose that there exists $i\ge 1$ such that $r_{2i}\ne 0$, and let $k$ be the minimal such $i$. By Part~2 of Theorem~\ref{theorem:DR hierarchy for F-CohFT}, we have
$$
P_1=P_1'+\eps^{2k+2}\alpha_{2k+2}u_{xx}^{k+1}+O(\eps^{2k+4}).
$$
Using Lemma~\ref{lemma:uniqueness}, we immediately obtain that the local functionals $\og_d$ have the form
$$
\og_d=\og'_d+\eps^{2k+2}\Delta\og_d+O(\eps^{2k+4}),\quad \Delta\og_d\in\Lambda_{u;2k+2},\quad d\ge 0.
$$
We also see that 
$$
K\frac{\delta\og_d}{\delta u}=\d_x\frac{\delta\og_d}{\delta u}+O(\eps^{2k+2}),\quad d\ge 0,
$$
which implies that $K=\d_x+O(\eps^{2k+2})$. So we have
$$
K=\d_x+\eps^{2k+2}\Delta K+O(\eps^{2k+3}),\quad \Delta K\in\cDO_u,\quad\deg\Delta K=2k+3.
$$
Since $K$ is a Poisson operator, the operator $\Delta K$ has the form
$$
\Delta K=L(f)\circ\d_x+\d_x\circ L(f)^\dagger,\quad f\in\cA_{u;2k+2}.
$$

\medskip

Let us now compute $\Delta\og_1$. The property $D_{\d_x P_1}(\og_1)=0$ implies that
$$
D_{\d_x(P_1'+\eps^{2k+2}\alpha_{2k+2}u_{xx}^{k+1})}\lb\og_1'+\eps^{2k+2}\Delta\og_1\rb=O(\eps^{2k+4}).
$$
Taking the coefficient of $\eps^{2k+2}$, we obtain the equation
$$
\alpha_{2k+2}D_{\d_x(u_{xx}^{k+1})}\lb\int\frac{u^3}{6}dx\rb+D_{uu_x}\lb\Delta\og_1\rb=0.
$$
By Lemma~\ref{lemma:uniqueness}, if a solution $\Delta\og_1$ of this equation exists, then it is unique. An elementary computation shows that 
$$
\Delta\og_1=\frac{\alpha_{2k+2}}{3k+2}\int uu_{xx}^{k+1}dx
$$
satisfies this equation.

\medskip

We have the equality
$$
\d_x\lb P_1'+\eps^{2k+2}\alpha_{2k+2}u_{xx}^{k+1}\rb=(\d_x+\eps^{2k+2}\Delta K)\frac{\delta}{\delta u}\lb \og_1'+\frac{\alpha_{2k+2}}{3k+2}\eps^{2k+2}\int u u_{xx}^{k+1}dx\rb+O(\eps^{2k+3}).
$$
Taking the coefficient of $\eps^{2k+2}$, we obtain
\begin{gather}\label{eq:equation for DeltaK-1}
\alpha_{2k+2}\d_x(u_{xx}^{k+1})=\frac{\alpha_{2k+2}}{3k+2}\d_x\frac{\delta}{\delta u}\int u u_{xx}^{k+1}dx+\Delta K\frac{\delta}{\delta u}\int\frac{u^3}{6}dx.
\end{gather}
We also know that
\begin{align}
&\Delta K=L(f)\circ\d_x+\d_x\circ L(f)^\dagger,\quad f\in\cA_{u;2k+2},\label{eq:equation for DeltaK-2}\\
&\frac{\d\Delta K}{\d u}=0.\label{eq:equation for DeltaK-3}
\end{align}
Let us prove that the system of these three equations~\eqref{eq:equation for DeltaK-1}--\eqref{eq:equation for DeltaK-3} for $\Delta K$ doesn't have a solution. For this, without loss of generality we can assume that $\alpha_{2k+2}=1$.

\medskip

Differentiating equation~\eqref{eq:equation for DeltaK-1} twice by $u$ and using that $\frac{\d\Delta K}{\d u}=0$, we obtain $\Delta K(1)=0$. Using~\eqref{eq:equation for DeltaK-2}, we get $\d_x\frac{\delta f}{\delta u}=0$, and therefore $\frac{\delta f}{\delta u}=0$. We conclude that $f=\d_x q$, for some $q\in\cA_{u;2k+1}$, and 
$$
\Delta K=\d_x\circ\lb L(q)-L(q)^\dagger\rb\circ\d_x.
$$

\medskip

Since $\frac{\d\Delta K}{\d u}=0$, we have $L(\frac{\d q}{\d u})=L(\frac{\d q}{\d u})^\dagger$, which implies that $\frac{\d q}{\d u}=\frac{\delta r}{\delta u}$ for some $r\in\cA_{u;2k+1}$. Choose $r'\in\cA_{u;2k+1}$ such that $\frac{\d r'}{\d u}=r$. Note that by changing $q\mapsto q-\frac{\delta r'}{\delta u}$ we don't change~$\Delta K$, but, on the other hand, we change $\frac{\d q}{\d u}\mapsto\frac{\d q}{\d u}-\frac{\delta r}{\delta u}=0$. So, without loss of generality, we can assume now that $\frac{\d q}{\d u}=0$.

\medskip

So we can now rewrite equation~\eqref{eq:equation for DeltaK-1} as follows:
$$
\d_x(u_{xx}^{k+1})=\frac{1}{3k+2}\d_x\frac{\delta}{\delta u}\int u u_{xx}^{k+1}dx+\lb\d_x\circ\lb L(q)-L(q)^\dagger\rb\circ\d_x\rb\lb\frac{u^2}{2}\rb,
$$
or equivalently
\begin{gather*}
u_{xx}^{k+1}=\frac{1}{3k+2}\frac{\delta}{\delta u}\int u u_{xx}^{k+1}dx+\lb L(q)-L(q)^\dagger\rb(uu_x),\quad q\in\cA_{u;2k+1},\quad\frac{\d q}{\d u}=0.
\end{gather*}
Taking the integral of both sides, we get
$$
\int u_{xx}^{k+1}dx=\frac{1}{3k+2}\int u_{xx}^{k+1} dx+\int L(q)(uu_x)dx,
$$
or equivalently
$$
\frac{3k+1}{3k+2}\int u_{xx}^{k+1}dx=\int D_{uu_x}(q) dx,
$$
but this equation doesn't have a solution, because of Lemmas~\ref{lemma:differentiating ulambda} and~\ref{lemma:unique density}. This contradiction proves that $r_{2i}=0$ for all $i\ge 1$.
\end{proof}

\medskip

\noindent{\bf Conclusions and future directions.} At least two natural questions arise from this work, which certainly deserve further investigation. The first one is how to prove the full DLYZ Conjecture \ref{conjecture:DLYZ}, or even the ALM Conjecture \ref{ALM-conjecture}. The second one is whether we can find a moduli space geometric origin for the functional parameters appearing in the ALM conjecture~\ref{ALM-conjecture}. Recall, indeed, that Theorem \ref{theorem:DR hierarchy for F-CohFT} only found such geometric origin for when the ALM functional parameters are constant. At least for the second question some ideas can be put forth on how to possibly enlarge the range of parameters that can be reached by the DR hierarchy construction, or slight generalizations thereof. Firstly, the DR construction works for the CohFTs and F-CohFTs without flat unit, a generalization which corresponds to removing the requirement that $\pi^* c_{g,n} = c_{g,n+1}$ in the definitions presented in Section \ref{section:CohFTs}. Rank $1$ CohFTs without flat unit are actually classified in \cite{Tel12} and correspond to exponentials in an infinite linear combinations of kappa classes. The corresponding DR hierarchy can be transformed into a deformation of the Riemann hierarchy by a simple change of coordinates which provide new deformations, although hardly as many as are allowed by all values of the functional parameters of the ALM conjecture. As for F-CohFTs (them too, possibly without unit), as remarked already in \cite{ABLR23}, one might argue that a (further) natural generalization consists in defining them on the partial compactification $\cM_{g,n}^{\mathrm{ct}}$ of $\cM_{g,n}$ given by curves of compact type, i.e. stable curves whose dual graph is a tree, instead of the full Deligne--Mumford compactification~$\oM_{g,n}$. Notice, indeed, that the axioms in the definition of F-CohFT of Section \ref{section:CohFTs} make sense on~$\cM_{g,n}^{\mathrm{ct}}$ as well, and that the DR hierarchy construction is anyway insensitive to the part of an F-CohFT that is supported on the complement of $\cM_{g,n}^{\mathrm{ct}}$. More in general, an exploration of the laxest set of axioms for families of classes on $\cM_{g,n}^{\mathrm{ct}}$ compatible with integrability might be in order.


\medskip

\begin{thebibliography}{BDGR18}

\bibitem[ABLR21]{ABLR21} A. Arsie, A. Buryak, P. Lorenzoni, P. Rossi. {\it Flat F-manifolds, F-CohFTs, and integrable hierarchies}. Communications in Mathematical Physics 388 (2021), 291--328.

\smallskip

\bibitem[ABLR23]{ABLR23} A. Arsie, A. Buryak, P. Lorenzoni, P. Rossi. {\it Semisimple flat F-manifolds in higher genus}. Communications in Mathematical Physics~397 (2023), 141--197.

\smallskip

\bibitem[ALM15a]{ALM15a} A. Arsie, P. Lorenzoni, A. Moro. {\it Integrable viscous conservation laws}. Nonlinearity~28 (2015), no.~6, 1859--1895.

\smallskip

\bibitem[ALM15b]{ALM15b} A. Arsie, P. Lorenzoni, A. Moro. {\it On integrable conservation laws}. Proceedings of the Royal Society~A~471 (2015), no.~2173, 20140124.

\smallskip

\bibitem[Bur15a]{Bur15-DR-hierarchy} A. Buryak. {\it Double ramification cycles and integrable hierarchies}. Communications in Mathematical Physics~336 (2015), no.~3, 1085--1107. 

\smallskip

\bibitem[Bur15b]{Bur15} A. Buryak. {\it Dubrovin--Zhang hierarchy for the Hodge integrals}. Communications in Number Theory and Physics~9 (2015), no.~2, 239--271.

\smallskip

\bibitem[Bur24]{Bur24} A. Yu. Buryak. {\it DR-hierarchies: from the moduli spaces of curves to integrable systems}. Proceedings of the Steklov Institute of Mathematics~325 (2024), 21--59.

\smallskip

\bibitem[BDGR18]{BDGR18} A. Buryak, B. Dubrovin, J. Guere, P. Rossi. {\it Tau-structure for the double ramification hierarchies}. Communications in Mathematical Physics 363 (2018), no.~1, 191--260.

\smallskip

\bibitem[BDGR20]{BDGR20} A. Buryak, B. Dubrovin, J. Guere, P. Rossi. {\it Integrable systems of double ramification type}. International Mathematics Research Notices 2020 (2020), no.~24, 10381--10446.

\smallskip

\bibitem[BGR19]{BGR19} A. Buryak, J. Guere, P. Rossi. {\it DR/DZ equivalence conjecture and tautological relations}. Geometry~\& Topology 23 (2019), no.~7, 3537--3600.

\smallskip

\bibitem[BR21]{BR18} A. Buryak, P. Rossi. {\it Extended $r$-spin theory in all genera and the discrete KdV hierarchy}. Advances in Mathematics 386 (2021), paper number 107794.

\smallskip

\bibitem[BSSZ15]{BSSZ15} A. Buryak, S. Shadrin, L. Spitz, D. Zvonkine. {\it Integrals of $\psi$-classes over double ramification cycles}. American Journal of Mathematics 137 (2015), no. 3, 699-737.

\smallskip

\bibitem[DMS05]{DMS05} L. Degiovanni, F. Magri, V. Sciacca.  {\it On deformation of Poisson manifolds of hydrodynamic type}. Communications in Mathematical Physics~253 (2005), no.~1, 1--24.

\smallskip

\bibitem[Dor78]{Dor78} I. Ya. Dorfman. {\it Formal variational calculus in the algebra of smooth cylindrical functions}. Functional Analysis and Its Applications~12 (1978), 101--107. 

\smallskip

\bibitem[Dub10]{Dub10} B. Dubrovin. {\it Hamiltonian PDEs: deformations, integrability, solutions}. Journal of Physics. A. Mathematical and Theoretical~43 (2010), no.~43, 434002.

\smallskip

\bibitem[DLYZ16]{DLYZ16} B. A. Dubrovin, S.-Q. Liu, D. Yang, Y. Zhang. {\it Hodge integrals and tau-symmetric integrable hierarchies of Hamiltonian evolutionary PDEs}. Advances in Mathematics 293 (2016), 382--435.

\smallskip

\bibitem[Get02]{Get02} E. Getzler. {\it A Darboux theorem for Hamiltonian operators in the formal calculus of variations}. Duke Mathematical Journal~111 (2002), no.~3, 535--560.

\smallskip

\bibitem[KM94]{KM94} M. Kontsevich, Yu. Manin. {\it Gromov--Witten classes, quantum cohomology, and enumerative geometry}. Communications in Mathematical Physics~164 (1994), no.~3, 525--562.

\smallskip

\bibitem[LRZ15]{LRZ15} S.-Q. Liu, Y. Ruan, Y. Zhang. {\it BCFG Drinfeld--Sokolov hierarchies and FJRW--Theory}. Inventiones Mathematicae 201 (2015), no. 2, 711--772.

\smallskip

\bibitem[LZ06]{LZ06} S.-Q. Liu, Y. Zhang. {\it On quasi-triviality and integrability of a class of scalar evolutionary PDEs}. Journal of Geometry and Physics~57 (2006), no.~1, 101--119.

\smallskip

\bibitem[LZ11]{LZ11} S.-Q. Liu, Y. Zhang. {\it Jacobi structures of evolutionary partial differential equations}. Advances in Mathematics~227 (2011), 73--130.

\smallskip

\bibitem[FP00a]{FP00-Hodge-and-Gromov} C.~Faber, R.~Pandharipande. {\it Hodge integrals and Gromov--Witten theory}. Inventiones Mathematicae~139 (2000), 173-199.

\smallskip

\bibitem[FP00b]{FP00} C. Faber, R. Pandharipande. {\it Logarithmic series and Hodge integrals in the tautological ring}. With an appendix by Don Zagier. Michigan Mathematical Journal 48 (2000), 215--252.

\smallskip

\bibitem[FP03]{FP03} C. Faber, R. Pandharipande. {\it Hodge integrals, partition matrices, and the $\lambda_g$ conjecture}. Annals of Mathematics~157 (2003), no.~1, 97--124.

\smallskip

\bibitem[PPZ15]{PPZ15} R. Pandharipande, A. Pixton, D. Zvonkine. {\it Relations on $\oM_{g,n}$ via $3$-spin structures}. Journal of the American Mathematical Society~28 (2015), no.~1, 279--309.

\smallskip

\bibitem[Ros17]{Ros17} P. Rossi. {\it Integrability, quantization and moduli spaces of curves}. SIGMA 13 (2017), article number~060.

\smallskip

\bibitem[Tel12]{Tel12} C. Teleman. {\it The structure of 2D semi-simple field theories}. Inventiones Mathematicae~188 (2012), no. 3, 525--588.

\end{thebibliography}
\end{document}